\documentclass[a4paper]{article}
\usepackage{amsmath,amsthm,amssymb,amsfonts, fancyhdr, color, comment, graphicx,environ,xcolor,enumitem,caption,subcaption,cancel,makecell,multicol,hyperref,float}
\usepackage[margin=1in]{geometry} 
\usepackage{authblk}
\usepackage[numbers]{natbib}
\usepackage{algorithm}
\usepackage{algpseudocode}
\usepackage{tabularx}
\usepackage{xcolor}
\usepackage{color}
\usepackage[table]{xcolor}
\usepackage{booktabs}
\usepackage{caption}
\usepackage{adjustbox}
\usepackage{makecell}
\usepackage{prodint}
\usepackage{tikz}
\usepackage{arydshln}
\usetikzlibrary{calc, positioning, arrows.meta}
\usepackage{multirow, booktabs}
\newcommand\citeay[1]{\citeauthor{#1} (\citeyear{#1})}

\makeatletter
\newcommand{\leqnomode}{\tagsleft@true}
\newcommand{\reqnomode}{\tagsleft@false}
\makeatother

\newtheorem{prop}{Proposition}
\newtheorem{assumption}{Assumption}

\makeatletter
\newcommand\pint{\DOTSI\if@display\PRODI\else\prodi\fi\ilimits@}
\makeatother

\newlength\myindent
\def\boxit#1{\vbox{\hrule\hbox{\vrule\kern6pt
          \vbox{\kern6pt#1\kern6pt}\kern6pt\vrule}\hrule}}
    \def\wt{\widetilde}

\newtheorem{Lem}{Lemma}
\newtheorem{Th}{Theorem}

\def\bse{\begin{eqnarray*}}
\def\ese{\end{eqnarray*}}
\def\be{\begin{eqnarray}}
\def\ee{\end{eqnarray}}
\def\E{{\mathbb E}}
\def\P{{\mathbb P}}

\def\wh{\widehat}
\def\wt{\widetilde}

\newcommand{\edit}[1]{{\color{blue}#1}} % all substantive revisions responding to TG's comments are wrapped in this

\title{Improving the efficiency of infectious disease prevention trials using negative control outcome event times}
\author[1,2]{Ethan Ashby}
\author[3]{Zhewei Zhang}
\author[4]{Tanya P. Garcia}
\author[1]{Ting Ye}
\author[2]{Holly Janes}
\author[2]{Bo Zhang}
\affil[1]{Department of Biostatistics, University of Washington, Seattle, WA, USA}
\affil[2]{Vaccine and Infectious Disease Division, Fred Hutchinson Cancer Center, Seattle, WA, USA}
\affil[3]{Department of Statistics, Penn State University, State College, PA, USA}
\affil[4]{Department of Biostatistics, Gillings School of Global Public Health, University of North Carolina at Chapel Hill, Chapel Hill, NC}

\begin{document}

\maketitle

\section*{Abstract}

Baseline covariate adjustment can potentially enhance the efficiency of randomized trials by improving precision of treatment effect estimates. However, the size of the precision gain depends on how strongly the baseline covariates are prognostic for the primary outcome. In randomized trials of infectious disease prevention interventions (e.g., vaccines or passively administered antibodies), a leading prognostic factor for infection --- an individual's exposure to the pathogen --- is rarely measurable at baseline, so conventional covariate adjustment yields only modest precision gains. We propose an alternative adjustment variable: a negative control outcome (NCO) event time, an infection that is unaffected by the intervention but associated with the primary outcome through overlapping unmeasured exposure mechanisms. Unlike a baseline covariate, a NCO event time is measured after randomization, which introduces a risk of selection bias if adjustment is not handled carefully. We formalize the causal assumptions under which adjustment for the NCO event time is valid, and show that right-censoring of the NCO event time further complicates adjustment. We derive the efficient influence function for the treatment-arm-specific survivor function of the primary outcome in a model where both the primary outcome and the NCO event time may be right-censored, and use it to construct a cross-fitted, one-step estimator that is multiply robust to nuisance misspecification and asymptotically efficient when the nuisances are estimated accurately at appropriate rates. In numerical experiments, our estimator matches the performance of unadjusted and baseline-covariate-adjusted estimators when the NCO event time is uninformative, and gains precision as the NCO event time becomes more prognostic for the primary outcome. We apply our method to HVTN 704/HPTN 085, a randomized, double-blinded trial of VRC01, a broadly neutralizing antibody against the HIV-1 virus. Adjusting for the time to a bacterial sexually transmitted infection --- a negative control outcome for HIV-1 acquisition --- reduced the estimated variance of the prevention efficacy estimate by approximately 27\%, compared to roughly 2.5\% for baseline covariate adjustment.
\section{Background}

Randomized trials are essential for establishing whether a biomedical
intervention --- a vaccine, a passively administered antibody, or another
prophylactic product --- causally effects a health outcome of interest
\citep{EMA2015_covariate_adjustment,FDA2021_covariate_adjustment_draft}.
Because trials can be expensive and time-consuming, there exists great interest in improving their efficiency.
One strategy to enhance trial efficiency is by improving the \textit{precision} of the intervention's estimated effect, as more precision can increase the power of the trial's primary analysis given a fixed sample size, or alternatively, enable a trial with smaller sample size while maintaining the desired power \citep{Yang_2001, tsiatis_covariate_2008, zhang_improving_2008, Lin_2013,
Ye_2022}. However, how much precision gain covariate adjustment buys depends on how strongly prognostic the baseline covariates are for the trial's primary outcome: adjusting for highly prognostic covariates can lead to substantial precision gains, while adjusting for weakly prognostic covariates can lead to negligible precision gain and even precision loss in finite samples \citep{ICH1998E9, van_lancker_covariate_2024}.

This dependence is especially limiting in randomized, placebo-controlled prevention trials for infectious disease, which test whether an intervention reduces infection or disease caused by a target pathogen. A leading prognostic factor for infection in these trials is an individual's effective level of \textit{exposure} to the pathogen --- the frequency and intensity of contact with an infectious source. Exposure can often be described qualitatively, but it is difficult or impossible to measure directly and incorporate into a statistical model \citep{farrington_2001, Farrington2013-lq}. In our case study, the pathogen of interest is Human Immunodeficiency Virus 1 (HIV-1) in populations of men who have sex with men (MSM), for which exposure is driven mainly by an individual's rate of sexual contact with virologically unsuppressed partners living with HIV-1 infection. Measuring HIV-1 exposure directly is difficult due to unreliability of self-reported sexual behavior (due to stigma and privacy considerations) and lack of objective data on contact frequency, prevention-method adherence, and partners' virological suppression statuses \citep{Schroder_2003}. Moreover, participant behaviors that drive HIV-1 exposure can change over the course of a trial, so even high-quality baseline measure of exposure may be only weakly predictive of exposure during trial follow-up. This difficulty is not unique to HIV-1: for community-transmitted respiratory pathogens such as SARS-CoV-2, exposure depends on a similarly unmeasurable rate of spatiotemporal proximity to infectious individuals in the community. Hence, exposure may fluctuate with seasonality, mobility, school and work attendance, and adherence to non-pharmaceutical interventions, making it essentially impossible to quantify directly. Because exposure itself cannot be measured, analysts fall back on adjusting for baseline covariates that are only weakly prognostic for infection, which often leads to modest precision gains in practice.

Are there other strategies that enable adjustment, at least in part, for exposure? One idea is to exploit correlation between infection outcomes with overlapping exposure mechanisms \citep{farrington_contact_2005,
farrington_correlated_2013}. For example, HIV-1 and other sexually transmitted infections (STIs) depend on similar sexual and behavioral characteristics, and positive ecological correlations between HIV-1 and other STI rates are well documented across contexts \citep{Looker2017-kj}, including in sub-Saharan Africa \citep{Kenyon2017-zq} and among MSMs (the population studied in our case study)
\citep{Baiers2024-dc, Stenger2021-jw, Grome2021-wz, Mullick2020-vv}. STIs have also been linked to a transient increase in the likelihood of HIV-1 acquisition caused by localized inflammation of the genital tract \citep{Craib1995-oo, Wu2021-lk}. The same idea applies to respiratory pathogens: SARS-CoV-2 and other community-transmitted pathogens (e.g., influenza and RSV) plausibly depend on a common set of factors driving exposure, which is supported by the frequency of co-infection \citep{Kim2020-wj} and by prior analyses of randomized trial data \citep{Ashby_2025}. \citeay{etievant_increasing_2022} used this idea
directly in the context of HPV vaccine trials, and proposed adjustment for binary indicators of vaccine-untargeted HPV strains to improve the precision of vaccine efficacy estimates against vaccine-targeted, carcinogenic HPV strains. This approach has two limitations that motivate our work. First, most infectious disease prevention trials record time-to-event infection outcomes, and reducing a full event time to a binary indicator discards information that could otherwise translate into additional precision. Second, an infection outcome measured after randomization risks selection bias unless it is causally unaffected by the intervention \citep{hernan_2020, Rosenbaum_1984}; we refer to such an infection outcome as a \textit{negative control outcome (NCO) event time} \citep{lipsitch_negative_2010, shi_selective_2020}. To justify adjusting for even a binary indicator of an NCO event time under right-censoring, \citeay{etievant_increasing_2022} required censoring to be completely unaffected by treatment, a stronger assumption than the one needed for a standard, unadjusted analysis using Kaplan-Meier curves.

An estimator that instead adjusts for the full NCO event time would recover the information the indicator discards, and could relax the completely-random-censoring assumption to something closer to what a standard survival analysis already requires. Nevertheless, building such an estimator is not straightforward, because both the primary outcome and the NCO event time are both subject to right-censoring. While many methods exist for right-censored outcomes, and a smaller set address a single right-censored covariate, very few address cases where the outcome and a key predictor are both right-censored. This represents a methodological gap with practical importance: prevention studies routinely collect data on infections caused by other pathogens, which represent potentially compelling markers of exposure to the target pathogen. However, how to use these other infection outcomes to improve the precision of treatment effect remains unsolved.

To address this gap, we propose an estimator of the treatment-arm-specific survivor function --- and so of prevention efficacy --- that adjusts for an NCO event time when both the NCO event time and the primary outcome are right-censored. Relative to \citeay{etievant_increasing_2022}, our estimator uses the full NCO event time rather than a binary indicator, so it retains information that binarization discards; it relaxes the completely-random-censoring assumption; and it is efficient and multiply robust to nuisance misspecification. In Section 2, we introduce the causal assumptions that license adjustment for the NCO event time, show why the natural approach --- adjusting for the observed, right-censored NCO event time directly --- can fail, and derive the efficient influence function (EIF) of the treatment-arm-specific survivor function when the NCO event time is right-censored. We then use the EIF to construct a cross-fitted, one-step estimator and establish that it is multiply robust and asymptotically efficient under regularity conditions. In Section 3, we compare our estimator to unadjusted and baseline-covariate-adjusted alternatives in numerical experiments designed to mimic randomized prevention trials. In Section 4, we apply our method to HVTN 704/HPTN 085, a randomized, double-blinded trial evaluating VRC01, a broadly neutralizing antibody, against HIV-1 acquisition in a population of MSM, using time to bacterial STI infection as the NCO event time.

\section{Methods}

\subsection{Notation and Data Structure}

\par Consider a randomized trial that collects baseline covariates $X$ (e.g., age) and randomly assigns a binary treatment $A$. The primary time-to-event outcome is $Y$. In infectious disease prevention trials, $Y$ could represent time-to-infection or illness caused by a target pathogen. Suppose the study also measures an NCO event time $N$ through routine surveillance or standard-of-care follow-up. In prevention trials, $N$ may denote time-to-infection or disease caused by a pathogen unaffected by the intervention.

\par To formalize the problem, we first introduce an ``oracle'' data model. The oracle setting represents an ideal scenario in which all relevant variables are fully observed for every participant, without censoring or missingness. In this setting, each participant contributes a complete data unit. We assume the oracle data units are drawn independent and identically distributed samples from a distribution $P_0$ in statistical model $M$.
\begin{align*}
    O^{\text{oracle}} = (A, X, Y, N) \overset{iid}{\sim} P_0 \in M
\end{align*}

In the developments below, we will frequently refer to the ``full" data unit, where only the primary outcome is right-censored, but the negative control infection time is fully observed.
\begin{align*}
    O^{\text{full}} = (A, X, \wt{Y}, \Delta_Y, N)
\end{align*}
Where $\wt{Y} := \min(Y, C_Y)$ and $\Delta_Y := I(Y \leq C_Y)$, where $C_Y$ is the censoring time for the primary outcome.

Both the oracle and full data units are never actually observed. Instead, we observe a coarsened data unit where $N$ is right-censored. Define the observed data unit as follows.
\begin{align*}
    O^{\text{obs}} = (A, X, \wt{Y}, \Delta_Y, \wt{N}, \Delta_N)
\end{align*}
Where $\wt{N} := \min(N, C_N)$ and $\Delta_N := I(N \leq C_N)$, where $C_N$ is the censoring time for the negative control outcome. 

In the sequel, we will distinguish between censoring due to dropout ($C$) and administrative censoring ($C_{\text{admin}}$), as administrative censoring is usually pre-specified in advance of the trial, and can be treated like a baseline covariate. We also do not impose restrictions on the censoring times $C_N$ and $C_Y$. If $Y$ and $N$ are subject to a common censoring time, $C_N=C_Y$ with probability 1. In other cases, when $N$ and $Y$ are measured at different time grids, $C_N$ and $C_Y$ may differ in arbitrary ways.

\subsection{Estimand, Causal Assumptions, and Identification}

We adopt Neyman and Rubin's potential outcomes framework \citep{Neyman_1921, Rubin_2005}, which presumes that each participant has a tuple of potential primary outcomes representing their event times under hypothetical assignment to placebo and treatment: $\{Y(0), Y(1)\}$. Let $S_{a}(t) = P_0(Y(a) > t)$ refer to the  causal treatment-arm-specific survivor function with respect to the primary infection. Our interest lies in identifying $S_a(t_0)$ for both $a \in \{0,1\}$ at some landmark time $t_0$, such that we can construct a contrast that quantifies the treatment's prevention efficacy. Herein, we focus on the causal log relative risk reduction:
\begin{equation}\label{eq:causalVE}
\begin{split}
    \log\text{RR}(t_0) &= \log\left(\frac{1-S_1(t_0)}{1-S_0(t_0)}\right).
\end{split}
\end{equation}
This parameter quantifies the reduction in cumulative incidence of the primary event at a prespecified, landmark time $t_0$. In infectious disease prevention trials, one minus the exponentiated version of \ref{eq:causalVE} corresponds to the prevention efficacy (PE) against the target pathogen on the cumulative incidence scale at time $t_0$.

To identify $S_{a}(t)$ and therefore the target parameter \ref{eq:causalVE} from observed data, we impose several causal assumptions. The assumptions can be summarized visually in the directed acylic graph (DAG) shown in Figure \ref{fig:DAGs}.

\begin{figure}[H]
    \centering
    \includegraphics[width=0.5\linewidth]{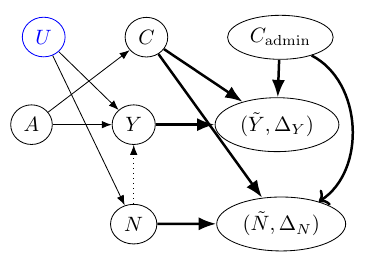}
    \caption{Graphical causal model illustrating assumptions underlying approach. $U$ is an unmeasured variable that is a common cause of both $Y$ and $N$ that is presumed to drive correlation between the infection types. The figure assumes a common dropout time $C=C_Y=C_N$. Arrows indicate causal relationships while while \textbf{bold} arrows denote deterministic relationships (implied by the deterministic functions $\min(x,y)$ and $I(x \leq y)$ used to obtain the minimum event time and the event indicator)}
    \label{fig:DAGs}
\end{figure}

\vspace{3mm}

\begin{assumption}[Consistency and Stable Unit Treatment Value Assumption (SUTVA)]
    Consistency requires that the observed outcome equals the potential outcome under the observed treatment assignment. SUTVA consists of the assumptions of (i) no interference between units, (ii) no multiple versions of treatment \citep{Rubin_1978}
\end{assumption}

Consistency of potential outcomes and no multiple versions of treatment are standard assumptions that are satisfied in most randomized prevention trials. In studies with infectious disease endpoints, the assumption of no interference between units may be violated due to spillover or herd immunity effects. This is a concern when the trial enrolls clusters of participants -- such as partners, households, or communities -- where the treatment status of one unit can affect the outcome of another unit. However, in individually randomized prevention trials, participants are embedded in a much larger target population, making interference between study participants very unlikely. Our case study considers an individually randomized prevention trial, where the impact of interference can be safely assumed to be negligible.

\vspace{3mm}

\begin{assumption}[Positive treatment probability]
\begin{align*}
    0 < \pi_0(a \mid x) \; \text{ for } a \in \{0,1\} \text{ and all } x \text{ with positive support} 
\end{align*}
\end{assumption}

Where $\pi_0(a \mid x) := P_0(A = a \mid X=x)$ is the conditional probability of receiving treatment $A=a$. Positivity implies that treatment is nondeterministic, which is satisfied by design in randomized trials.

\vspace{3mm}

\begin{assumption}[Strong ignorability \citep{Rubin_1978}]
\begin{align*}
    A \perp (Y(0), Y(1), N(0), N(1), C_N(0), C_N(1), C_Y(0), C_Y(1), X)
\end{align*}
\end{assumption} 
Where $(N(0), N(1))$ are the potential outcomes for the NCO event time, and $(C_Y(1), C_Y(0))$ and $(C_N(0), C_N(1))$ are the vector of potential censoring times for $Y$ and $N$ respectively. This assumption is satisifed by design in randomized controlled trials by virtue of randomization of treatment.

\vspace{3mm}

Assumptions 1-3 are sufficient to identify $S_a(t_0)$ using the Oracle data unit, $O^{\text{oracle}}$. However, when primary event times are right-censored, additional assumptions on the censoring mechanism are required to identify the target parameter from $O^{\text{full}}$ and $O^{\text{obs}}$.

\vspace{3mm}

\begin{assumption}[Independent censoring]
\begin{align*}
    &(C_Y,C_N) \perp (Y,N) \mid (A, X) \\
    &P_0(C_Y \geq t_0 \mid A,X) \geq \epsilon > 0 \qquad P_0(C_N \geq \min(\wt{N}, t_0) \mid A,X,\wt{Y},\Delta_Y) \geq \epsilon > 0
\end{align*}
\end{assumption} 
The first part of Assumption 4 assumes that the censoring times are independent of the event times conditional on baseline factors, an assumption frequently made in survival analysis. The latter case ensures that censoring is not exhaustive within subgroups of participants at the relevant evaluation points: at the landmark time $t_0$ for $C_Y$ and at each participant's observed $\wt{N} \wedge t_0$ for $C_N$. In settings where $N$ and $Y$ share a common censoring mechanism (i.e., $C_Y=C_N$), when $\Delta_Y=0$, $P_0(C_N \geq t \mid a, x, \wt{y}, \delta_y) = I(t \leq \wt{y})$, which still satisfies positivity at the relevant evaluation point. Some important consequences of Assumption 4 are established in Appendix \ref{prf:rem1}.

While not required for identification, we introduce Assumptions 5 and 6 on the NCO event time that unlock potential efficiency gains. In plain terms, Assumption 5 requires that the NCO event time carry no causal effect of treatment (so that adjusting for it cannot introduce selection bias). Assumption 6 requires that the NCO event time is informative about the unmeasured exposure process that impacts the primary outcome (so that adjusting for it can actually improve precision).

\vspace{3mm}

\begin{assumption}[Negative control outcome (NCO)]\label{ass:NCO}
\begin{align*}
    N = N(0) = N(1) \hspace{3mm} \text{with probability } 1
\end{align*}
\end{assumption} 

\par Assumption 5 requires that $N$ satisfy the sharp causal null hypothesis of no treatment effect. By precluding an effect of $A$ on $N$, we can adjust for $N$ without incurring selection bias \citep{Rosenbaum_1984}. A consequence of the NCO assumption is statistical independence between $N$ and $A$ conditional on any other variable(s). Next, we introduce an assumption that is sufficient to describe cases where adjusting for $N$ can lead to precision gain. 

\begin{assumption}[Approximate U-comparability \citep{shi_selective_2020}]\label{ass:Ucomparable}
    Let $U_Y$ be an unmeasured cause of $Y$ and $U_N$ be an unmeasured cause of $N$. Formally, assume $Y(a)$ and $N$ are generated from the nonparametric structural equation model (NPSEM),
    \begin{align*}
        Y(a) &= h_Y(U_Y, A=a, X, N, \epsilon_Y) \\
        N &= h_N(U_N, X, \epsilon_N)
    \end{align*}
    where $h_Y, h_N$ are unknown deterministic function and $\epsilon_Y \perp \epsilon_N$ are completely random errors. If the unmeasured common causes for $Y$ and $N$ are associated, meaning $U_Y \not \perp U_N \mid A, X$, then $N$ is correlated with $U_Y$ conditional on the other measured variables.
    \begin{align*}
        N \not \perp U_Y \mid A, X
    \end{align*}
\end{assumption}

Figure \ref{fig:DAGs} illustrates the special case where unmeasured variables $U_Y$ and $U_N$ are identical. In infectious disease prevention trials, $U_N$ and $U_Y$ may not be unequal but correlated through similar exposure pathways. As depicted on the DAG in Figure \ref{fig:DAGs}, the structural equation for $N$ in Assumption \ref{ass:NCO} does not depend on the treatment variable $A$. The structural equation for $Y(a)$ \textit{may} depend on $N$ but not vice versa, otherwise $N$ may be influenced by treatment indirectly via its affect on $Y$.

Following \citeay{Westling_2023}, Assumptions 1-5 are sufficient to identify $S_a(t)$ from the \textit{full data}, $O^{\text{full}}$, using the G-formula applied to the conditional treatment-arm-specific survivor function.
\begin{align*}
    S_{a_0}(t) &= \mathbb{E}_{O^{\text{full}}}[S_{0Y}(t \mid a_0, X, N)]
\end{align*}
Where $S_{0Y}(t \mid a_0, X, N) := P_0(Y > t \mid A=a_0, X, N)$.

Next, we introduce the following proposition, proven in Appendix \ref{prf:prop1}, which establishes the identification of $S_{a_0}(t_0)$ from the \textit{observed data} where $N$ may be right-censored. In short, our proposition shows that targeting the true NCO event time through inverse probability weighting (IPW), not the censored event time directly, achieves identification under our assumptions. This identification approach lets use use the full time-to-event information in $N$ rather than only its censoring status.
\begin{prop}\label{prop:id}
    Let $O^{\text{obs}} := (A, X, \wt{Y}, \Delta_Y, \wt{N}, \Delta_N)$. The following G-formulas that attempt to adjust for the coarsened versions of $N$, $\wt{N}$ and $\Delta_N$ respectively, do not identify $S_{a_0}(t)$ due to post-treatment selection.
    \begin{align*}
        \mathbb{E}_{O^{\text{obs}}}\left[P(Y > t \mid A=a_0, X, \wt{N})\right] &\neq S_{a_0}(t) \\
        \mathbb{E}_{O^{\text{obs}}}\left[P(Y > t \mid A=a_0, X, \Delta_{N})\right] &\neq S_{a_0}(t)
    \end{align*}
    Rather, we can identify $S_{a_0}(t)$ using inverse probability weighting (IPW) of the complete-$N$-cases.
    \begin{align*}
        \mathbb{E}_{O^{\text{obs}}}\left[\frac{\Delta_N I(A=a_0)}{G_{0N}(\wt{N} \mid A, X, \wt{Y}, \Delta_Y) \pi_0(a_0 \mid X)} S_{0Y}(t \mid a_0, X, N=\wt{N}) \right] &=S_{a_0}(t) 
    \end{align*}
\end{prop}

In short, adjusting for observed NCO event indicators or event times as suggested by \citet{etievant_increasing_2022} can fail to identify $S_a(t)$ if censoring is impacted by treatment \citep{Rosenbaum_1984}. Assuming censoring times are completely unaffected by treatment may be plausible if participants are blinded and the treatment does not impact study dropout. However, this assumption is hard to justify through scientific knowledge and is a stronger assumption than required for a standard Kaplan-Meier analysis. 

In the next section, we will use this IPW identification approach in \ref{prop:id} as the basis for constructing an estimating equation for $S_{a_0}(t_0)$ with favorable theoretical properties.

\subsection{Efficiency calculations}

In this section, our goal is to derive of the efficient influence function (EIF) of the treatment-arm-specific survival probability $S_{a_0}(t_0)$ in the observed data model. We focus on the EIF because it characterizes the minimum asymptotic variance achievable among all regular asymptotic linear (RAL) estimators in a semiparametric model. The EIF can also be used to construct the most efficient estimator of $S_{a_0}(t_0)$ among all RAL estimators. Then, inference on $S_{a}(t_0)$ for $a \in \{0,1\}$ can be translated to inference on the target parameter $\log \text{RR}(t_0)$ using the delta method.

To derive the EIF in the observed data model when $N$ is right-censored, we adopt a projection approach. As a starting point, following Theorem 1 of \citeay{Westling_2023}, the EIF of $S_{a_0}(t_0)$ in the full-data model where $Y$ is right-censored but $N$ is completely observed is given by the following,
\begin{equation}\label{eq:full_EIF}
\resizebox{0.93\textwidth}{!}{
$\begin{aligned}
    \phi^F_{a_0,t_0}(o^{\text{full}}) &= S_{0Y}(t_0 \mid a_0, x, n) \left[1 - \frac{I(a=a_0)}{\pi_0(a_0 \mid x)} \left\{\frac{I(\wt{y} \leq t_0, \delta_y=1)}{S_{0Y}(\wt{y} \mid a, x, n) G_{0Y}(\wt{y} \mid a, x)} - \int_0^{\min(t_0, \wt{y})} \frac{\Lambda_{0Y}(du \mid a, x, n)}{S_{0Y}(u \mid a, x, n) G_{0Y}(u \mid a, x)}\right\}\right]
\end{aligned}$}
\end{equation}
Where $S_{0Y}(t \mid a, x, n) := P_0(Y > t \mid A=a, X=x, N=n)$ is the conditional survival function for the primary outcome, $\Lambda_{0Y}(u \mid a, x, n)$ is the corresponding cumulative hazard of $Y$, $G_{0Y}(t \mid a, x) := P_0(C_Y \geq t \mid A=a, X=x)$ is the conditional survival function of $C_Y$, and $\pi_0(a\mid x) := P_0(A=a \mid X=x)$ is the treatment propensity score.

In reality, $N$ is also subject to right-censoring. So instead of the full data unit, we observe $O^{\text{obs}} = (A,X,\wt{Y},\Delta_Y, \wt{N}, \Delta_N)$ where $\wt{N}:=\min(N, C_N)$ and $\Delta_N := I(N \leq C_N)$. Note we will make no assumption about the joint relationship between $C_N$ and $C_Y$. They may coincide when $N$ and $Y$ are measured on a common visit grid and are subject to the same dropout process, they may be deterministically related, or they differ in arbitrary ways. 

Before we present the observed-data efficient influence function, we define two key nuisance parameters. The first is the conditional survival function of $C_N$:
\begin{align*}
    G_{0N}(t \mid a,x,\wt{y},\delta_y) := \exp\left\{- \int_0^t \lambda_{C_N}(u \mid a, x, \wt{y}, \delta_Y) du\right\}.
\end{align*}
This parameter is identifiable from observed data based on the following exclusion restriction implied by Assumption 4 (see Appendix \ref{prf:rem1} for details).
\begin{equation}\label{rem:indep_1}
    C_N \perp N \mid (A,X,\wt{Y},\Delta_Y).
\end{equation}
We note that when both $Y$ and $N$ are subject to the same censoring time ($C=C_N=C_Y$), $G_{0N}$ is a derived quantity from $G_{0Y}$ rather than a separate nuisance. In this case, $G_{0N}(\cdot \mid a, x, \wt{y}, \delta_y=0)$ yields a point mass at $\wt{y}$ while when $G_{0N}(\cdot \mid a, x, \wt{y}, \delta_y=1) \equiv G_{0N}(\cdot \mid a, x, C>\wt{y})$ yields a left-truncated survival function.

The second nuisance parameter is the conditional mean of the full-data EIF, conditional on all the other observed variables.
\begin{align*}
    Q_0(u \mid a, x, \wt{y}, \delta_y) := \mathbb{E}[\phi^F_{a_0,t_0}(\wt{y}, \delta_y, a, x, N) \mid N \geq u, A=a, X=x, \wt{Y}=\wt{y}, \Delta_Y=\delta_y]
\end{align*}

Below, we present the observed-data EIF, which is obtained by applying the theory of influence function projections onto the tangent space of coarsened data models \citep{Tsiatis2006-cc} (see Appendix \ref{prf:Thm1} for a proof). In plain words, the observed-data EIF combines two sources of information: for participants whose NCO event time is fully observed ($\Delta_N=1$), it uses the full-data EIF directly; for participants whose NCO event time is censored ($\Delta_N=0$), it uses the partial information contained in $Q_0$, the conditional mean of the full-data EIF given what is observed, integrated against the martingale for the NCO censoring process. This is what allows the estimator to remain efficient even though the NCO event time --- the variable driving the precision gain --- is not always fully observed.

\begin{Th}[Observed-data EIF]\label{th:EIF_deriv}
Let $\phi^F_{a_0,t_0}$ denote the full-data EIF as shown in \ref{eq:full_EIF}. Under censoring at random (Assumption 4), the observed-data influence function corresponding to $O^{\text{obs}} = (A,X,\wt{N},\Delta_N, \wt{Y}, \Delta_Y)$ is given by the following.
\begin{equation}\label{eq:AIPWCC_IF}
    \begin{split}
    &\phi^{\text{obs}}_{a_0,t_0}(o^{\text{obs}}) = \frac{\delta_{n}}{G_{0N}(\wt{n} \mid a,x,\wt{y},\delta_y)} \phi^F_{a_0,t_0}(\wt{y},\delta_{y},a,x,\wt{n}) + \int_0^{\infty} \frac{Q_0(u \mid a, x, \wt{y}, \delta_y)}{G_{0N}(u \mid a,x,\wt{y},\delta_y)} dM_{C_N}(u \mid a, x, \wt{y}, \delta_y)
    \end{split}
\end{equation}
Where $dM_{C_N}(u \mid a, x, \wt{y}, \delta_y) := dN_{C_N}(u) - \lambda_{C_N}(u \mid a, x, \wt{y}, \delta_y) I(N \geq u, C_N \geq u) du$ is the canonical martingale increment for the censoring variable $C_N$, where $dN_{C_N}(u) = I(C_N = u, N > u)$ is the jump process for NCO censoring and $\lambda_{C_N}(u \mid a, x, \wt{y}, \delta_y)$ is the NCO censoring hazard. 

Moreover, because $\phi^F_{a_0, t_0}$ is the nonparametric efficient influence function in the full-data model and $\phi^{\text{obs}}_{a_0,t_0}$ is its unique projection onto the observed-data tangent space, $\phi^{\text{obs}}_{a_0,t_0}$ is necessarily the observed-data efficient influence function.
\end{Th}

Our interest in the observed-data efficient influence function (EIF) stems from its role as a optimal estimating function for the target parameter $S_{a_0}(t_0)$. Next, we introduce an estimating-equation-based estimator of $S_{a_0}(t_0)$ built atop the EIF along with its theoretical properties. 

\subsection{A cross-fitted one-step estimator}

As written, $\phi_{a_0,t_0}^{\text{obs}}$ depends on at least four nuisance parameters: the treatment propensity $\pi_0(a \mid x)$, the conditional survival function of $Y$-censoring $G_{0Y}(t \mid a, x)$, the conditional density of the irrelevant infection outcome $f_{0N}(n \mid a, x, \wt{y}, \delta_y)$, and the conditional survival function of the primary outcome $S_{0Y}(y \mid a, x, n)$. If $N$ and $Y$ are subject to different censoring processes, we need to estimate an additional nuisance parameter, the conditional survival function of $C_N$: $G_{0N}(t \mid a, x, \wt{y}, \delta_y)$. Suppose temporarily that we have a means of estimating all these quantities. Let $\wh{\phi}^{\text{obs}}_{n, a_0, t_0}$ denote the value of $\phi^{\text{obs}}_{a_0, t_0}$ when the nuisance estimators are plugged in for the unknown true values. The classic one-step estimator is $\mathbb{P}_n \wh{\phi}^{\text{obs}}_{n, a_0, t_0}$, where $\mathbb{P}_n$ denotes the empirical probability measure (i.e., sample average) of the EIF values \citep{Tsiatis2006-cc, Bickel_semi_1993}. In this case, the one-step estimator is also an estimation equations-based estimator because the influence function is linear in the target parameter.

The asymptotic properties of RAL estimators, including the one-step estimator, depend on the nuisance parameters in two important ways. First, the nuisance estimators must converge sufficiently quickly to their true values to control so-called second-order remainder terms. Second, the nuisance estimators must fall within sufficiently small function classes to ensure negligibility of empirical process terms. These criteria appear to be in conflict: simple parametric models for nuisance functions are unlikely to converge to their true values, which motivates the use of more flexible, data-adaptive estimators. However, flexible nuisance estimators are unlikely to fall in sufficiently simple function classes, leading to failure to control empirical process terms. A well-established approach to resolve this conflict is to combine data-adaptive nuisance estimation with cross-fitting, which eliminates the complexity constraint on the nuisances \citep{Bickel1982-mu, Chernozhukov2017-mh, chernozhukov_2018}.

We define a cross-fitted one-step estimator for $S_{a_0}(t_0)$ as follows. Following \citet{Westling_2023}, for a deterministic integer $2 \leq K \leq \lfloor n/2 \rfloor$, we randomly divide participants $\{1, \ldots, n\}$ into $K$ disjoint folds $\mathcal{V}_{n,1}, \ldots, \mathcal{V}_{n,K}$ with cardinalities $n_1, \cdots, n_K$. We require that these sets be of as close to equal sizes as possible, so that $|n_k-n/K|\le 1$ for each $k$, and that the number of folds $K$ be bounded as $n$ grows. For each $k \in \{1, \ldots, K\}$, we define $\mathcal{T}_{n,k} := \{O_i^{\text{obs}} : i \not \in \mathcal{V}_{n,k}\}$ as the training set for fold $k$. Define $\wh{\phi}^{\text{obs}}_{n,k,a_0,t_0}$ as the EIF where the unknown nuisance parameters are replaced by estimates calculated exclusively from the training set $\mathcal{T}_{n,k}$ , which includes all observations not in the $k$-th fold $\mathcal{V}_{n,k}$. The cross-fitted one-step estimator is
\begin{equation}\label{eq:cfosest}
    \wh{S}_{n, a_0}(t) := \frac{1}{n} \sum_{k=1}^K \sum_{i \in \mathcal{V}_{n,k}} \wh{\phi}^{\text{obs}}_{n,k,a_0,t_0}(O_i^{\text{obs}}).
\end{equation}

Below, we introduce two important theoretical results for the cross-fitted, one-step estimator. The theoretical results use the following regularity conditions:
	\renewcommand{\labelenumi}{(C\arabic{enumi})}
	\renewcommand{\theenumi}{C\arabic{enumi}}
	\begin{enumerate}
		\item\label{assump:c1} There exists $\pi_\infty$, $G_{\infty,Y}$, $G_{\infty,N}$, $S_{\infty,Y}$, and $f_{\infty,N}$ such that:
		\bse
		&&\max_k\ \mathbb E\left\{\frac{1}{\wh \pi_{n,k}(a_0\mid X)}-\frac{1}{\pi_\infty(a_0\mid X)}\right\}^2 \overset{p}{\to} 0,\\
		&&\max_k\ \E\left\{ \sup_{u\in [0,t_0]}\left|\frac{1}{\wh G_{n,k,Y}(u\mid a_0,X)}-\frac{1}{G_{\infty,Y}(u\mid a_0,X)}\right| \right\}^2\overset{p}{\to} 0,\\
		&&\max_k\ \E\left\{\sup_{u\in[0,t_0]}\left|\frac{\wh S_{n,k,Y}(t_0\mid a_0,X)}{\wh S_{n,k,Y}(u\mid a_0,X)}-\frac{S_{\infty,Y}(t_0\mid a_0,X)}{S_{\infty,Y}(u\mid a_0,X)}\right|\right\}^2\overset{p}{\to} 0,\\
		&& \max_k\ \E\left\{\sup_{u\in[0,\tau_N]}\left|\frac{1}{\wh G_{n,k,N}(u\mid a_0,X,\wt Y,\Delta_Y)}-\frac{1}{G_{\infty,N}(u\mid a_0,X,\wt Y,\Delta_Y)}\right|\right\}^2\overset{p}{\to}0,\\
		&&\max_k\ \E\left\{\sup_{u\in[0,\tau_N]}\left|\wh Q_{n,k}(u\mid a_0,X,\wt Y,\Delta_Y)-Q_{\infty}(u\mid a_0,X,\wt Y,\Delta_Y)\right|\right\}^2\overset{p}{\to}0,
		\ese
        where $\tau_N>0$ is a fixed constant, and the nuisance functions
        $G_N$ and $Q$ are estimated and evaluated on the interval
        $[0,\tau_N]$.

		\item\label{assump:c2} There exists a constant $\varepsilon>0$ such that, with probability tending to one, for almost every relevant covariate value, $	\wh\pi_{n,k}(a_0\mid x)\ge \varepsilon$, $\pi_\infty(a_0\mid x)\ge \varepsilon$, $\wh G_{n,k,Y}(t_0\mid a_0,X)\ge \varepsilon$, $G_{\infty,Y}(t_0\mid a_0,x)\ge \varepsilon$, $\wh G_{n,k,N}(u\mid a_0,x,\wt y,\delta_y)\ge \varepsilon$, and $G_{\infty,N}(u\mid a_0,x,\wt y,\delta_y)\ge \varepsilon$.
	\end{enumerate}

Condition \ref{assump:c1} requires convergence of the nuisance estimators to fixed limit functions, which is used to control empirical process terms. Condition \ref{assump:c2} is a positivity assumption ensuring that the treatment and censoring probabilities remain uniformly bounded away from zero. 

In plain terms, the first theoretical result establishes that the cross-fitted one-step estimator remains consistent even when some of its nuisance functions are estimated incorrectly. This property is referred to as multiple robustness. Formally, we require only a subset of the nuisance estimators to be consistent for the true unknown parameters, which is formalized in Condition \ref{assump:c3} below. 

\begin{Th}[Multiple robustness]\label{th:MR}
Suppose Assumptions 1-6 and regularity Conditions \ref{assump:c1}-\ref{assump:c2}  hold. Moreover, suppose that the at least one of the following subsets of nuisance estimators are consistent for the true values. Letting a $\infty$-subscript refer to an estimator's probability limit, suppose at least one of the four conditions below are satisifed.
{\leqnomode
\begin{equation*}\label{assump:c3}
\tag{C3}
\begin{aligned}
(S_{\infty, Y},G_{\infty, N}) = (S_{0Y},G_{0N}), \qquad
(S_{\infty, Y},f_{\infty, N}) = (S_{0Y}, f_{0N}), \qquad
\\
(\pi_{\infty},G_{\infty, Y},G_{\infty, N}) = (\pi_0, G_{0Y}, G_{0N}) \qquad
(\pi_{\infty},G_{\infty, Y},f_{\infty, N}) = (\pi_{0},G_{0Y},f_{0N}).
\end{aligned}
\end{equation*}}
Then the cross-fitted one-step estimator $\wh S_{n,a_0}(t_0)$ is consistent for $S_{a_0}(t_0)$. Thus, consistency does not require all nuisance functions to be correctly specified simultaneously.
\end{Th}

The next theorem establishes the conditions under which the cross-fitted, one-step estimator achieves the minimum possible variance among all RAL estimators in the semiparametric model. Importantly, the theorem provides a means of quantifying the variance of the cross-fitted, one-step estimator using the variance of the influence function. The ensuing theorem is proven in \ref{prf:3}. The key condition, Condition \ref{assump:c4} (displayed below), is what permits the use of flexible, data-adaptive nuisance estimators without jeopardizing consistency and asymptotic normality of the final estimator. requires roughly that the rates of convergence of $(S_{n,Y}-S_{0Y})(\pi_n-\pi_0)$, $(S_{n,Y}-S_{0Y})(G_{n,Y}-G_{0Y})$, and $(f_{n,N}-f_{0N})(G_{n,N}-G_{0N})$ to zero be faster than $n^{-1/2}$. This is the key condition that permits the use of flexible, machine-learning nuisance estimators without jeopardizing consistency. The ensuing theorem is proven in \ref{prf:2}.

\begin{Th}[Asymptotic linearity]\label{Th:AL}
Suppose Assumptions 1-6 and regularity Conditions \ref{assump:c1}-\ref{assump:c3} hold. Also, assume Condition C4 holds, which claims that the nuisance estimators converge quickly enough to their true values, such that the following second-order remainder terms are controlled. 
\[
\left\{
\begin{aligned}\label{assump:c4}
r_{a_0,t_0,1}
&\equiv
\max_k
\E\Bigl|
\{\widehat{\pi}_{n,k}(a_0)-\pi_0(a_0)\}
\{\widehat{S}_{n,k,Y}(t_0)-S_{0Y}(t_0)\}
\Bigr|
= o_P(n^{-1/2}),
\\[1ex]
r_{a_0,t_0,2}
&\equiv \max_k \E\Biggl|
\widehat{S}_{n,k,Y}(t_0)
\int_0^{t_0}
\left\{
\frac{G_{0Y}(u)}{\widehat{G}_{n,k,Y}(u)}-1
\right\}
\left(
\frac{S_{0Y}(u)}{\widehat{S}_{n,k,Y}(u)}-1
\right)
(du \mid a_0,X,N)
\Biggr| = o_P(n^{-1/2})
\rlap{\qquad\textnormal{(C4)}},
\\[1ex]
r_{a_0,t_0,3}
&\equiv \max_k \E\Biggl|
\int_0^N
\{\widehat{Q}_{n,k}(u)-Q_k(u)\}
\{d\widehat{\Lambda}_{n,k,C_N}(u)-d\Lambda_{0C_N}(u)\}
\Biggr| = o_P(n^{-1/2}).
\end{aligned}
\right.
\]
Under this additional condition, then
\bse
\sqrt n\{\wh S_{n,a_0}(t_0)-S_{a_0}(t_0)\} =
\frac{1}{\sqrt n}\sum_{i=1}^n\left\{
\phi^{\mathrm{obs}}_{a_0,t_0}(O_i^{\mathrm{obs}})
- S_{a_0}(t_0) \right\} + o_p(1).
\ese
Consequently,
\[
\sqrt n\{\widehat S_{a_0}(t_0)-S_{a_0}(t_0)\}
\rightarrow N(0,\sigma_0^2)
\]
in distribution, where
\[\sigma_0^2 = E_0\left[\{\phi^{\mathrm{obs}}_{a_0,t_0}(O^{\mathrm{obs}})-S_{a_0}(t_0)\}^2\right].
\]
\end{Th}

Together, these two results develop the asymptotic theory underlying the use of the cross-fitted, one-step estimator adjusting for the NCO event time for inference on the treatment-arm-specific survival function for the primary outcome.

\subsection{Nuisance parameter specification and estimation}

This section discusses estimation of the nuisance functions that the EIF and cross-fitted, one-step estimator depend on. We primarily focus on modeling two nuisance functions that together implicitly define a joint distribution of the primary and NCO event times, which poses several challenges.

The first nuisance function --- the treatment propensity score $\pi_0(a\mid x)$ --- is known by design in a randomized trial and can therefore be estimated consistently using logistic regression. We advocate estimating the remaining nuisances using cross-fitted SuperLearner ensembles \citep{vanderLaanetal2007,Westling_2023}, which combine predictions from a library of base learners. For example, we recommend estimating the censoring survival functions $G_{0Y}(t\mid a,x)$ and $G_{0N}(t\mid a,x,\widetilde Y,\Delta_Y)$ with SuperLearner ensembles containing parametric (e.g., Weibull), semiparametric (e.g., Cox and Generalized Additive Cox models), and nonparametric survival models (e.g., random survival forests).

The main focus of this section involves estimating the remaining two nuisance functions, which each describe the conditional distribution of one event time given the other:
\[
f_{0N}(n\mid a,x,\widetilde{Y},\Delta_Y)
\quad\text{and}\quad
S_{0Y}(y\mid a,x,n).
\]

One challenge is that estimating these quantities separately can lead to estimated conditional distributions that do not correspond to any valid joint distribution of $(Y,N)\mid(A,X)$, a phenomenon referred to as incompatibility \citep{Arnold1989-xx,Chen2010-as}. To avoid incompatibility, we instead model the joint distribution through the one-way factorization of the joint likelihood
\[
f_N(n\mid a,x) \times S_{0Y}(y\mid a,x,n).
\]
Parameterizing the joint distribution in this way guarantees compatibility by design. However, it is yet unclear how the one-way factorization helps specify a model for one of the target nuisances, $f_{0N}(n\mid a,x,\widetilde{Y},\Delta_Y)$. In the proposition below, we show that this desired nuisance function can be derived from the components of the one-way factorization using Bayes rule (see Appendix \ref{prf:prop2} for a proof).

\begin{prop}[Deriving target nuisance under one-way factorization]\label{prop:one_way}
    We can express the target nuisance $f_{0N}(n \mid a, x, \wt{y}, \delta_y)$ in terms of the one-way factorization $f_{N}(n \mid a, x)$ and $S_{0Y}(y \mid a, x, n)$ via Bayes rule.
    \begin{align*}
    f_{0N}(n \mid a, x, \wt{y}, \delta_y) &= f_N(n \mid a, x) \cdot \frac{L(\wt{y}, \delta_y \mid a, x, n)}{L(\wt{y}, \delta_y \mid a, x)}
    \end{align*}
    where
    \begin{align*}
    \frac{L(\wt{y}, \delta_y \mid a, x, n)}{L(\wt{y}, \delta_y \mid a, x)} &=\left[\delta_y \left(\frac{f_{0Y}(\wt{y} \mid a, x, n)}{f_{Y}(\wt{y} \mid a, x)}\right) + (1-\delta_y) \frac{S_ {0Y}(\wt{y} \mid a, x, n)}{S_Y(\wt{y} \mid a, x)}\right].
    \end{align*}
    Note that $S_Y(u \mid a, x) = \int_0^{\infty} S_{0Y}(u \mid a, x, n) f_{N}(n \mid a, x) dn$ by tower law. $f_{0Y}$ and $f_{Y}$ refer to the conditional density functions associated with $S_{0Y}$ and $S_Y$, which can be obtained from the survival functions using the following identity $f(x) = - (d/du) S(u) \mid_{u=x}$. 
\end{prop}

Hence, $f_{0N}(n \mid a, x, \wt{y}, \delta_y)$ can be written as $f_{N}(n \mid a, x)$ tilted by the likelihood ratio of $Y$ induced by conditioning on $(N,A,X)$ relative to conditioning on $(A,X)$ alone. When $(Y \perp N \mid A, X)$, then $L(\wt{y}, \delta_y \mid a, x, n)/L(\wt{y}, \delta_y \mid a, x) =  1$ and $f_{0N}(n \mid a, x, \wt{y}, \delta_y) = f_{N}(n \mid a, x)$. When $Y$ and $N$ are associated through unmeasured causes, the tilt deviates from 1 and re-weights values of $N$ that make the observed data $(a,x,\wt{y}, \delta_y)$ more likely. 

We note that expressing $f_{0N}(n \mid a, x, \wt{y}, \delta_y)$ through the one-way factorization induces dependence between $S_{0Y}$ and $f_{0N}$, so that these nuisance functions are no longer variation independent. As a result, not all nuisance-function combinations appearing in the multiple robustness result of Theorem \ref{th:MR} can be achieved independently under this parameterization. Consequently, the set of nuisance configurations yielding consistency may be smaller than that stated in Theorem \ref{th:MR}. Nevertheless, the resulting estimator remains doubly robust and retains the same rate-double-robustness property required for asymptotic linearity in Theorem \ref{Th:AL}.

The remaining challenge is how to model $S_{0Y}(y \mid a, x, n)$, which is complicated by the fact that the outcome $Y$ and key regressor $N$ may be only partially observed due to censoring. A naive strategy is to restrict estimation to ``NCO-complete'' cases (i.e., $\Delta_N=1$). However, this can induce selection bias whenever $N \not\perp Y \mid (A,X)$, since the complete cases need not be representative of the target population. In addition, this approach discards information on the primary outcome ($Y$), which is particularly inefficient when event rates are low.

To address this issue, we propose an EM algorithm that alternates between filling in information about the right-censored NCO event time and obtaining interim estimates of key nuisance parameters. The algorithm is depicted visually in Figure \ref{fig:tables_with_arrow}.

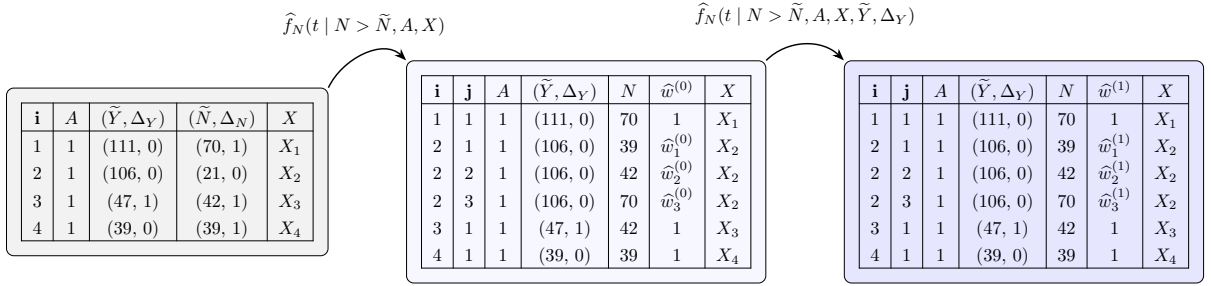
\begin{figure}[ht]
\setlength{\extrarowheight}{3pt}
    \centering
    \resizebox{\textwidth}{!}{
    \begin{tikzpicture}[
        % Define a style for the tables as TiKZ nodes
        table_node/.style={
            draw, 
            inner sep=0pt, 
            align=center
        },
        % Define a style for the arrows
        arrow_style/.style={
            -{Stealth}, 
            thick, 
            black
        }
    ]

    % Node 1: First table

    \node[
    draw,
    rounded corners,
    fill=gray!10,
    inner sep=8pt](TableA){
    \begin{tabular}{|c|c|c|c|c|}
            \hline
            \textbf{i} & $A$ & $(\wt{Y}, \Delta_{Y})$ & $(\wt{N}, \Delta_{N})$ & $X$ \\
            \hline
            1 & 1 & (111, 0) & (70, 1) & $X_1$\\
            2 & 1 & (106, 0) & (21, 0) & $X_2$\\
            3 & 1 & (47, 1) & (42, 1) & $X_3$\\
            4 & 1 & (39, 0) & (39, 1) & $X_4$\\
            \hline
        \end{tabular}
    };

    % Node 2: Second table, positioned to the right of the first
    \node[
    draw,
    rounded corners,
    fill=blue!03,
    right=1.5cm of TableA,
    inner sep=8pt] (TableB) {
        \begin{tabular}{|c|c|c|c|c|c|c|}
            \hline
            \textbf{i} & \textbf{j} & $A$ & $(\wt{Y}, \Delta_{Y})$ & $N$ & $\wh{w}^{(0)}$ & $X$ \\
            \hline
            1 & 1 & 1 & (111, 0) & 70 & 1 & $X_1$\\
            2 & 1 & 1 & (106, 0) & 39 & $\wh{w}^{(0)}_1$ & $X_2$\\
            2 & 2 & 1 & (106, 0) & 42 & $\wh{w}^{(0)}_2$ & $X_2$\\
            2 & 3 & 1 & (106, 0) & 70 & $\wh{w}^{(0)}_3$ & $X_2$\\
            3 & 1 & 1 & (47, 1) & 42 & 1 & $X_3$\\
            4 & 1 & 1 & (39, 0) & 39 & 1 & $X_4$\\
            \hline
        \end{tabular}
    };

    % Node 2: Second table, positioned to the right of the first
    \node[
    draw,
    rounded corners,
    fill=blue!10,
    right=1.5cm of TableB,
    inner sep=8pt] (TableC) {
        \begin{tabular}{|c|c|c|c|c|c|c|}
            \hline
            \textbf{i} & \textbf{j} & $A$ & $(\wt{Y}, \Delta_{Y})$ & $N$ & $\wh{w}^{(1)}$ & $X$ \\
            \hline
            1 & 1 & 1 & (111, 0) & 70 & 1 & $X_1$\\
            2 & 1 & 1 & (106, 0) & 39 & $\wh{w}^{(1)}_1$ & $X_2$\\
            2 & 2 & 1 & (106, 0) & 42 & $\wh{w}^{(1)}_2$ & $X_2$\\
            2 & 3 & 1 & (106, 0) & 70 & $\wh{w}^{(1)}_3$ & $X_2$\\
            3 & 1 & 1 & (47, 1) & 42 & 1 & $X_3$\\
            4 & 1 & 1 & (39, 0) & 39 & 1 & $X_4$\\
            \hline
        \end{tabular}
    };

    % Draw the arrow between the nodes
    % Connects the right side of TableA to the left side of TableB
    \draw[arrow_style](TableA.north east)to[out=65, in=135] node[
    midway,
    above=8pt,
    fill=white,
    inner sep=2pt]{$\displaystyle \wh f_N(t \mid N>\wt N, A, X)$}(TableB.north west);
    
    \draw[arrow_style]
    (TableB.north east)
    to[out=45, in=135]
    node[
    midway,
    above=8pt,
    fill=white,
    inner sep=2pt]
{$\displaystyle \wh f_N(t \mid N>\wt N, A, X, \wt Y, \Delta_Y)$}
(TableC.north west);
    
    \end{tikzpicture}}
    \caption{A visual depiction of the EM algorithm used to estimate $S_{0Y}(y \mid a, x, n)$.}
    \label{fig:tables_with_arrow}
\end{figure}

The first stage of the EM algorithm (the E-step) replaces observations with right-censored NCO event times with complete NCO event times using a dataset augmentation strategy inspired by the redistribute-to-the-right algorithm \citep{Efron_1967, Tanner1987-ex, taylor_survival_2002}. The augmentation expands each NCO-censored observation ($\Delta_N=0$) into a collection of pseudo-observations corresponding to observed event times exceeding the censoring time. Then, each pseudo-observation is assigned a weight that redistributes the original observation’s mass across the augmented support of $N$. In the first pass, we use working redistribution weights based on the conditional distribution $N \mid (A,X)$, which is obtained from the first component of the one-way factorization:
\[
\wh w_j^{(0)} = \wh{f}_{N}(N_{j} \mid A, X, N_{j}>\wt{N}) = 
\frac{\wh f_N(N_j \mid A,X)}{\wh S_N(\wt N \mid A,X)}.
\]

In the second stage of the EM algorithm (the M-step), we estimate the desired nuisance $S_{0Y}$ via a redistribution-weighted regression of $Y$ on $(A,X,N)$ on the augmented dataset. A first pass of the weighted regression using $\wh w^{(0)}$ yields a pilot estimator $\wh{S}_{Y}^{(0)}$. Because the initial redistribution weights are based only on the conditional distribution $N \mid (A,X)$, they do not exploit the information about NCO event time $N$ contained in the primary outcome $Y$. Consequently, a single EM iteration is generally insufficient. Instead, we can update the redistribution weights using the full conditional distribution $N \mid (A,X,Y)$, after which the E- and M-steps are iterated until convergence. Concretely, we can update the redistribution weights for NCO-censored observations using the current estimate of the posterior distribution of $N\mid(A,X,\wt Y,\Delta_Y)$:
\[
\wh w_j^{(m+1)} \propto
\wh L^{(m)}(\wt Y,\Delta_Y \mid A,X,N_j) \cdot
\wh f_N(N_j \mid A,X, N_j > \wt{N}).
\]
Where $\wh{L}^{(m)}$ is the estimated likelihood defined above based on the current estimate of the conditional survival function $\wh{S}^{(m)}_{Y}$ (see \ref{prop:one_way}). Then we obtain an updated estimator $\wh{S}_{0Y}^{(m+1)}$ by fitting a weighted regression of $Y$ on $(A,X,N)$ using updated weights $\wh w^{(m+1)}$. This procedure is iterated until convergence or until the relative change in the weights satisfies
\[
d_{\mathrm{rel}}^{(m)} :=
\left(
\frac{\sum_{j\in\mathcal J} (\wh w_j^{(m+1)} - \wh w_j^{(m)})^2}
{\sum_{j\in\mathcal J} (\wh w_j^{(m)})^2}
\right)^{1/2}
< \gamma,
\]
for some tolerance threshold $\gamma$ and where $\mathcal{J}$ indexes the pseudo-observations requiring redistribution.

\section{Numerical Experiments}

We simulated data from randomized trials of $n=\{2400, 4000\}$ participants with $\pi=0.5$ probability of randomization to intervention/placebo. We considered trials where the intervention's prevention efficacy (PE) against the primary infection was $25\%$ ($\text{logRR}=\log\{0.75\}$). We assumed that the study administratively censored all infection outcomes at two years in accordance with the HVTN 704/HPTN 085 study \citep{corey_two_2021}. We will summarize the performance of estimators of the log risk ratio ($\log\{\text{RR}\}$) based on the mean relative bias, median relative efficiency (compared to unadjusted estimator), and mean empirical confidence interval coverage of at the largest primary event time.

We simulated $\{Y(1), Y(0), C(1), C(0), N, X\}$ from a multivariate normal copula model. We assumed intercorrelation between $Y(1)$, $Y(0)$, and $N$ of $\rho \in \{0, 0.5, 0.8\}$, reflecting cases where the NCO event time was unassociated, moderately associated, and highly associated with primary infection times. We assumed the potential dropout times $C(0)$ and $C(1)$ were not intercorrelated with each other nor all other time-to-event variables, thereby satisfying the censoring at random assumption. We considered cases where $X$ was uncorrelated with all time-to-event variables ($\rho_X = 0$). We assumed that $Y(a)$ had a marginal exponential distribution with annual incidence of either $6\%$ in the placebo arm ($\lambda_{Y(a)} = -\log(1-0.06 \cdot [1+a \cdot \text{RR}])$). We assumed that $C(a)$ had a marginal exponential distribution with annual incidence of $\lambda_C = 10\%$ in both arms, reflecting instances of moderate to high censoring. We assumed that $N$ had a marginal exponential distribution with annual incidence of $12\%$ ($\lambda_N = -\log(1-0.12)$). We assumed a standard normal marginal distribution of baseline covariate $X$.

As benchmarks, we considered estimators of $\log \text{RR}(t_0)$ based on unadjusted, na\"ive Kaplan-Meier curves and estimators that only adjusted for the fully observed, weakly prognostic baseline covariate $X$ \citep{Westling_2023}. We evaluated two versions of our estimator. The first was an ``Oracle" estimator which replaced the unknown nuisance parameters with true values, which were computed using highly stratified approximations computed from an external dataset with 5 million observations generated from the same model as the trial data. The latter version was the proposed cross fitted (CF) estimator, which used a logistic regression model to estimate $\pi_0(a \mid x)$ and used cross-fitted SuperLearner ensembles to estimate all other nuisances.\citep{vanderLaan_2007, Westling_2023} Specifically, we used 6-fold cross-fitted SuperLearner ensembles to estimate $f_{N}(n \mid a, x)$, and $G_{0Y}(t \mid a, x)$ with Kaplan-Meier, Exponential, Cox, and generalized additive Cox models as base learners. To estimate $S_{0Y}(y \mid a, x, n)$, we used a 6-fold cross-fitted SuperLearner with Cox, generalized additive Cox, and random forest base learners applied to the augmented dataset with redistribution weights as described above. We conducted iterative EM-style updates of our estimator until the augmented dataset weight vectors updates were small ($d_{\text{rel}}^{(m)} < 10 ^{-4}$) or 5 iterations were reached, whichever came first.

The results of the simulation are shown in Table \ref{tab:sim_rho_2}. All estimators showed low bias ($<2.5\%$ average relative bias) for the target parameter and achieved near nominal confidence interval coverage. As $N$ became more prognostic for $Y$, adjusting for $N$ led to greater efficiency gain. When $N$ was moderately prognostic for $Y$ ($\rho=0.5$), the efficiency gain was $8-9\%$. When $N$ was highly prognostic for $Y$, the efficiency gain was approximately $33\%$ and relative bias decreased. Hence, our proposed estimators exhibited strong performance on par with the benchmark estimators when $N$ is not prognostic, but enhance efficiency when $N$ is prognostic for $Y$. In other words, adjusting for the NCO event time costs nothing when it turns out to be uninformative, and buys precision gain that scales with how strongly the NCO event time is associated with the primary outcome.

\begin{table}[ht]
\centering
\caption{Simulation results at landmark time $t_0$ based on 1000 replications.}
\label{tab:sim_rho_2}
\resizebox{\textwidth}{!}{\begin{tabular}{llccccccccc}
\toprule
& & \multicolumn{3}{c}{$\rho = 0.0$} 
& \multicolumn{3}{c}{$\rho = 0.5$} 
& \multicolumn{3}{c}{$\rho = 0.8$} \\
\cmidrule(r){3-5} \cmidrule(r){6-8} \cmidrule(r){9-11}
$n$ & Method
& Rel. Bias (\%) & Rel. Eff. & Cov (\%) 
& Rel. Bias (\%) & Rel. Eff. & Cov (\%) 
& Rel. Bias (\%) & Rel. Eff. & Cov (\%) \\
\midrule
\multirow{4}{*}{$2000$} & KM & 1.680 & 1.000 & 94.9 & 1.540 & 1.000 & 94.1 & 2.230 & 1.000 & 94.4\\
& X-adj & 1.690 & 1.000 & 95.1 & 1.520 & 1.000 & 93.9 & 2.190 & 1.000 & 94.3\\
& NCO-adj (Oracle) & 1.710 & 0.995 & 95.1 & 2.080 & 0.911 & 93.9 & 0.528 & 0.669 & 94.7\\
& NCO-adj (CF) & 1.760 & 0.998 & 94.8 & 1.920 & 0.922 & 94.4 & 0.425 & 0.689 & 94.2\\
\\[1ex]
\multirow{4}{*}{$4000$} & KM & -0.480 & 1.000 & 95.4 & -0.574 & 1.000 & 94.7 & 0.338 & 1.000 & 94.4\\
& X-Adj & -0.534 & 1.000 & 95.6 & -0.569 & 1.000 & 94.7 & 0.324 & 1.000 & 94.4\\
& NCO-adj (Oracle) & -0.588 & 0.997 & 95.4 & -0.721 & 0.912 & 94.6 & -0.124 & 0.671 & 94.2\\
& NCO-adj (CF) & -0.581 & 0.998 & 95.0 & -0.715 & 0.921 & 93.9 & -0.119 & 0.687 & 93.6\\
\bottomrule
\end{tabular}}
\end{table}

\section{Case Study: a randomized trial of a broadly neutralizing antibody for HIV-1 prevention}

In the following section, we apply our methods to the HVTN 704/HPTN 085 (AMP) study \citep{corey_two_2021}, which was a randomized, double-blinded trial evaluating the prevention efficacy of a broadly neutralizing antibody (bNab), VRC01 against HIV-1 infection among 2701 men who have sex with men and transgender women in North America, South America, and Europe. HVTN 704/HPTN 085 Participants were randomized in a 1:1:1 ratio to receive low-dose VRC01, high-dose VRC01, and placebo atop standard of prevention. The primary outcome was serologically confirmed HIV-1 infection at 80 weeks of follow up, and HIV-1 testing was conducted every 4 weeks. Participants also underwent sexually transmitted infection testing at enrollment and every 6 months or after presentation of symptoms. Sexually transmitted infection testing included gonorrhea (oropharyngeal, urethral, and rectal), chlamydia (urethral and rectal), and syphilis (serologic) \citep{edupuganti_feasibility_2021}.

As a first step, we will evaluate if each STI is a valid NCO event time for HIV-1 infection. Second, we will explore whether adjustment for time-to-STI can improve the precision of the treatment effect of VRC01 against HIV-1. We focused on two candidate NCO event times---time-to-first rectal gonorrhea (RGC) and syphilis infections---as these were previously suggested as factors independently associated with higher likelihood of HIV-1 \citep{Craib1995-oo, Vaughan2015-mx, Barbee2017-ak, Wu2021-lk}.

\subsection{Rectal Gonorrhea and Syphilis are valid negative controls for HIV-1 in HVTN 704/HPTN 085}

We empirically checked if RGC and other bacterial STIs were valid NCO event times for HIV-1 (Assumption 5) by comparing estimated cumulative incidence functions between the placebo and pooled VRC01 arms (low dose plus high dose). We fit cumulative incidence functions separately by randomization arm while adjusting for age (in years), racial/ethnic underrepresented minority (URM) status, and a baseline HIV-1 risk score using the method of \citeay{Westling_2023} to improve precision and relax assumptions by allowing independent censoring conditional on covariates.

As shown in Figure \ref{fig:ncovalid}, we observe that there is no discernible effect of VRC01 administration on incidence of RGC over 72 weeks. Moreover, we observed that binary RGC infection status was associated with higher cumulative incidence of HIV-1 infection over follow-up. In particular, placebo recipients with an incident RGC infection appeared at approximately 50\% higher relative risk of acquiring HIV-1 over 80 weeks compared to placebo recipients who avoided RGC infection. These results supported that RGC was a compelling NCO event time for HIV-1, as it was both unaffected by treatment yet positively associated with HIV-1. Results were similar for syphilis, another STI previously proposed as a marker of HIV-1 exposure. Hence, both outcomes appeared to promising choices of NCO event times that could plausibly lead to precision gain in the evaluating whether VRC01 could prevent HIV-1 infection.

\begin{figure}[H]
 \centering
  \begin{subfigure}[b]{0.8\textwidth}
    \includegraphics[width=\textwidth]{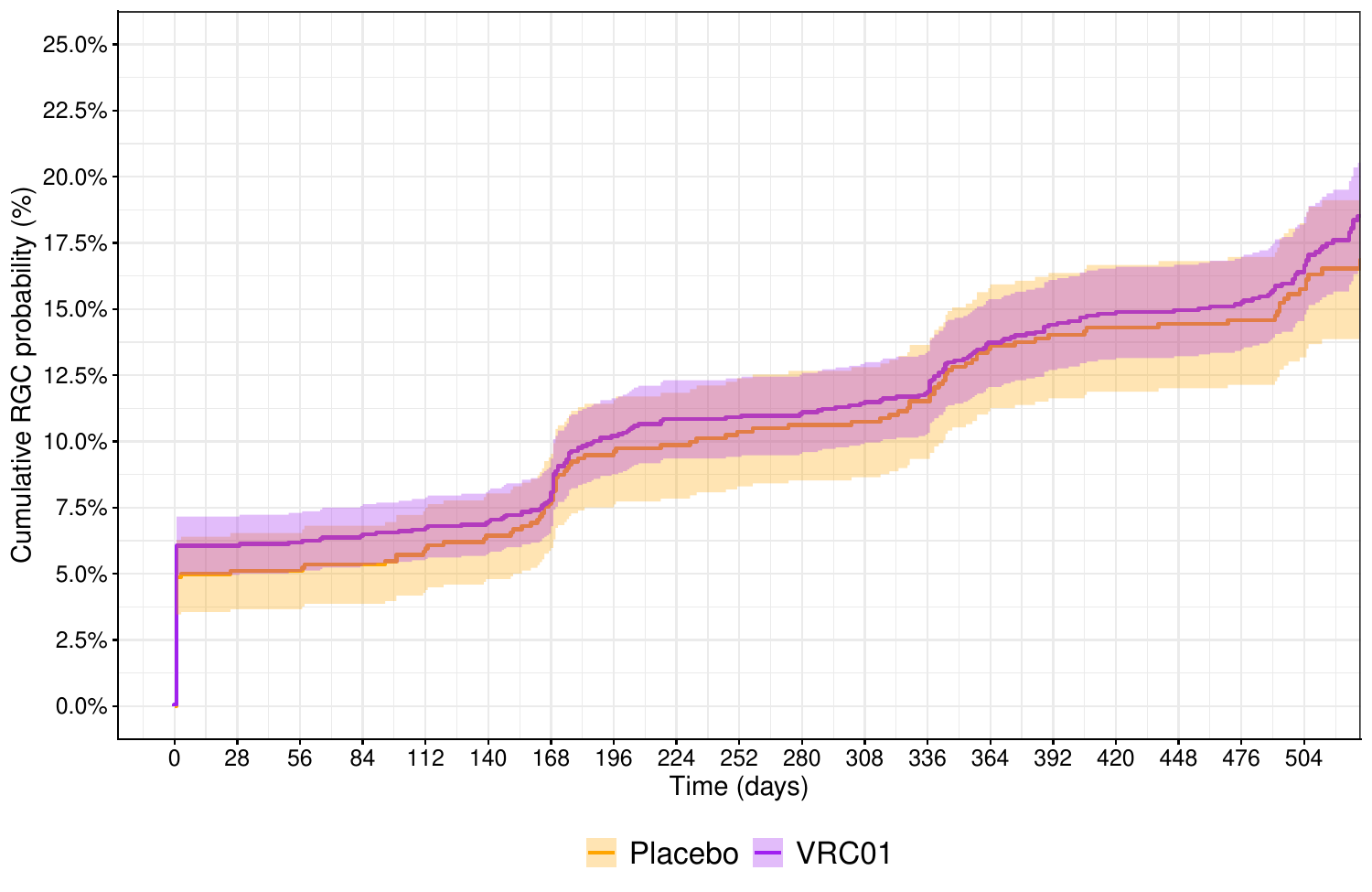}
    \label{fig:fig1}
  \end{subfigure}
  \hfill
  \begin{subfigure}[b]{0.85\textwidth}
    \includegraphics[width=\textwidth]{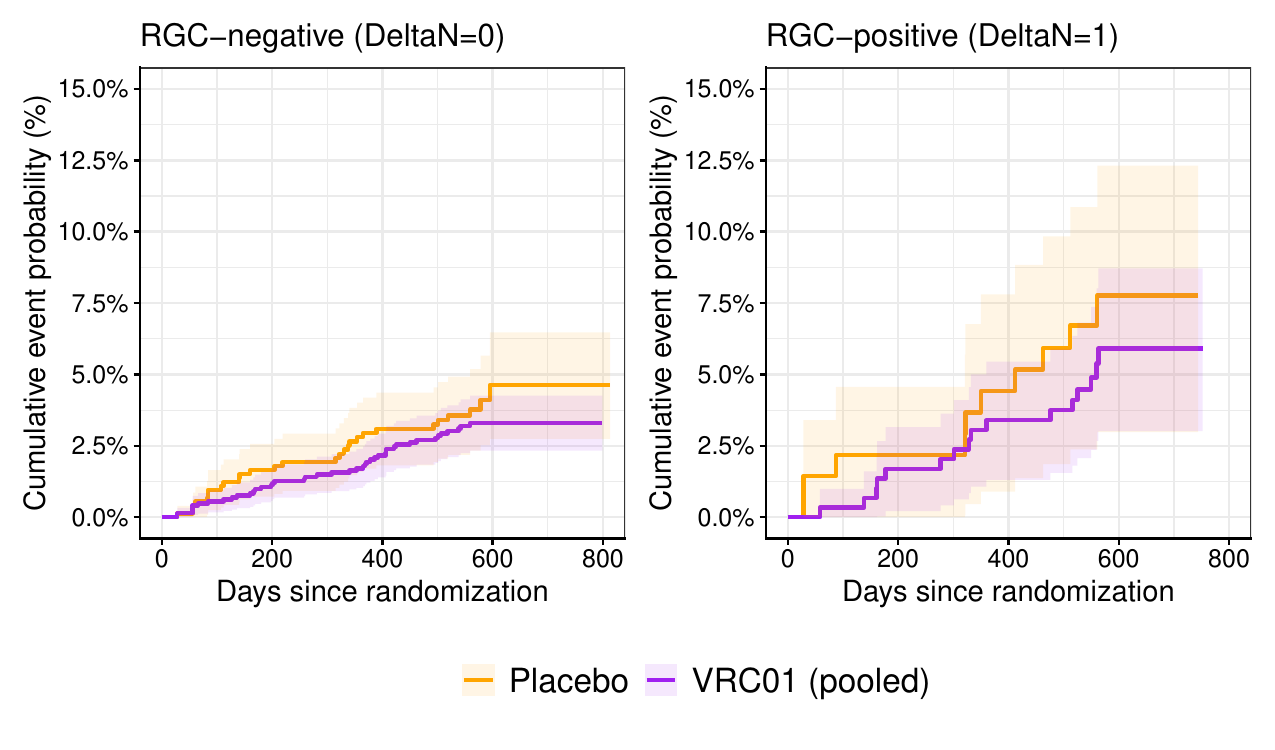}
    \label{fig:fig2}
  \end{subfigure}
  \caption{Top: cumulative incidence functions for Rectal Gonorrhea (RGC) over 72 weeks as a function of randomization to placebo or VRC01 (low or high dose) among the modified intention to treat efficacy cohort with pointwise confidence intervals. Right: ratio of cumulative incidence functions of HIV-1 over 80 weeks of follow-up stratified by RGC-infection status ($\Delta_N=0$ or $\Delta_N=1$).}
    \label{fig:ncovalid}
\end{figure}

\subsection{Adjustment for STIs improves precision in estimated prevention efficacy of VRC01 against HIV-1}

We investigated whether adjustment for time-to-incident STIs could improve precision of prevention efficacy estimates of VRC01 in HVTN 704/HPTN 085. As a benchmark, we estimated the logRR and PE using (i) an unadjusted estimator based on Kaplan-Meier estimators and (ii) a one-step estimator as described in \citeay{Westling_2023} which adjusted for a baseline HIV-1 risk score estimated using machine learning \citep{corey_two_2021}. We compared the benchmark estimators to an estimator adjusting for time-to-first rectal gonorrhea (RGC) as the key adjustment variable. We used a logistic regression model to estimate the treatment propensity and an 8-fold cross-fitted SuperLearner to estimate each remaining nuisance. The density $f_{N}(n \mid a, x)$ was estimated using an ensemble of exponential, Weibull, Cox, and additive Cox models. The survival function $S_{0Y}(y \mid a, x, n)$ was estimated using an ensemble of Cox, additive Cox, and survival random forest. We used the EM algorithm based on data augmentation to estimate $S_{0Y}$ with a maximum of $5$ iterations of the EM algorithm per fold. All censoring distributions were estimated using an ensemble of Kaplan-Meier, Exponential, Weibull, Cox, additive Cox, and random forest base learners. All outcomes were summarized at the week 80 visit, consistent with the primary analysis.

All estimators supported that the pooled VRC01 group had moderately lower incidence of HIV-1 compared to the placebo arm. While adjusting for the baseline risk score ($X$) only achieved a modest 2.5\% reduction in the estimated variance of the target estimand, our proposed estimator which adjusted for time-to-first RGC infection exhibited a 26.9\% variance reduction relative to the unadjusted analysis. Similarly, adjusting for time-to-first syphilis infection reduced the estimated variance by 26.5\%. Hence, while baseline covariate adjustment led to very small precision gains, adjusting for RGC and Syphilis as NCO event times led to much greater precision gains.

\begin{table}[h]
    \centering
    \begin{tabular}{ccccc}
        \hline
        Estimator & logRR & Est. Var & Relative Eff. & PE ($95\%$ CI) \\
        \hline
        Kaplan-Meier & -0.309 & 4.60e-2 & 1 & $26.6\% (-11.7\%, 51.8\%)$ \\
        $X$-adjusted & -0.294 & 4.49e-2 & 0.975 & $25.4\% (-12.9\%, 50.8\%)$\\
        RGC-adjusted & -0.248 & 3.36e-2 & 0.731  & $22.0\% (-11.7\%,45.6\%)$\\
        Syph-adjusted & -0.261 & 3.38e-2 & 0.735 & $23.0\% (-10.5\%,46.3\%)$\\
        \hline
    \end{tabular}
    \caption{Comparison of different estimators of VRC01 prevention efficacy in HVTN 704/HPTN 085. Estimates are reported at Week 80.}
    \label{tab:HVTN 704/HPTN 085_results}
\end{table}

\section{Discussion}

We proposed an estimator of the treatment-arm-specific survivor function --- and so of prevention efficacy --- that adjusts for a negative control outcome (NCO) event time when both the NCO event time and the primary outcome are right-censored. We derived the efficient influence function for this doubly-right-censored setting and used it to construct a cross-fitted, one-step estimator that is multiply robust and asymptotically efficient, together with an EM algorithm based on data augmentation and redistribute-to-the-right for the hardest of the required nuisance functions. In simulations, the estimator was consistent and well-calibrated, and its precision gain grew with how strongly the NCO event time was associated with the primary outcome.

Applying our method to the HVTN 704/HPTN 085 study, we observed that adjusting for time-to-first rectal gonorrhea or time-to-first syphilis infection reduced the estimated variance of VRC01's efficacy estimate by roughly 27\%, compared to a 2.5\% reduction from the best available baseline covariate. This practical result suggests that adjusting for infections caused by pathogens unrelated to the vaccine's target could lead to more precise and powerful assessments of prevention efficacy without adjusting the trial's design, enrollment, or budget.

Our work contributes to the growing statistical literature on regression with censored covariates,\citep{May2011-yh, Wang2012-yq, Bernhardt2015-no, Atem2016-zb, Atem_2017, Ding2018-oe, Atem_2019, Matsouaka_2020, Chen2022-lx, Ashner2024-sh, Lotspeich2024-wm} particularly those that view the problem from the lens of semiparametric theory in coarsened data models \citep{Bickel_semi_1993, Robins1994-hv, Tsiatis2006-cc, Lee2024-hi, Vazquez2025-kq, Zhang2025-vy, Vazquez2024-dk}. To our knowledge, our work is the first to consider the problem of a right-censored adjustment variable and a right-censored primary outcome jointly, in the context of infectious disease prevention trials. Our approach extends that of \citeay{Etievantetal2023} which proposed adjusting for binary indicators of observed ``vaccine untargeted" HPV infections in HPV vaccine trials. However, our approach more naturally accommodates the typical time-to-event nature of infection endpoints in prevention trials, avoids information loss by binarization, and eschews the need for a completely random censoring assumption.

Two limitations of our approach are worth stating plainly, alongside what the method still delivers in each case. First, the size of the precision gain is a property of the trial's design rather than of the estimator: trials with short follow-up or infrequent, non-symptom-prompted NCO testing will see smaller gains, because there is less information in the NCO event time to exploit. Even in that setting, however, our simulations show the estimator performs no worse than an unadjusted analysis when the NCO assumption is satisifed, which suggests that adjusting for the NCO event time costs nothing in large sample practical problems. Second, the negative control assumption --- that the NCO event time is unaffected by the intervention --- is a strong assumption and adjustment can lead to bias if it fails to hold. However, the NCO assumption makes a falsifiable prediction that can be checked directly in the data: no effect of the intervention on incidence of the candidate NCO event time. We checked this prediction for HVTN 704/HPTN 085 and found no discernible effect of VRC01 on rectal gonorrhea incidence. Moreover, the antigen-specific nature of the immune response elicited by interventions like vaccines and passively administered antibodies provides very strong biological rationale for the NCO assumption. In our view, the combination of a testable prediction and mechanistic plausibility makes the NCO assumption reasonable in many cases.

In future work, it may be of interest to extend our framework to accommodate multivariate and/or recurrent right-censored NCO endpoints. In our developments, we assume a single, time-to-first NCO event is measured in parallel to the outcome of interest. However, in HVTN 704/HPTN 085, various STIs such as gonorrhea (oropharyngeal, urethral, and rectal), chlamydia (urethral and rectal), and syphilis (treponemal and non-treponemal serology) were measured throughout follow-up. Each STI outcome is a noisy proxy of factors that influence an individual's HIV-1 exposure, so including multiple STI outcomes in the analysis may unlock additional efficiency gain. Moreover, many STIs do not produce sterilizing immunity, so persistent/recurrent infections can occur. Chronic STIs can transiently increase risk of HIV-1 transmission and acquisition \citep{Chun2013-rj, Cohen2019-ld}. Accounting for persistent or recurrent STIs may provide additional explanatory information about HIV-1 exposure beyond the time-to-first STI, which may further boost precision and power. Finally, adjusting for quantitative readouts of STI tests (e.g., viral/bacterial load) is a yet unexplored idea that could potentially offer additional efficiency gains.

\section*{Acknowledgments and Funding Statement}

The authors thank the AMP study participants and study investigators. We would also like to thank Peter Gilbert and Larry Corey (Fred Hutchinson Cancer Center) for their editorial input on this paper.

This project has been funded in whole or in part by the National Institutes of Health, National Institute of Allergy and Infectious Diseases (R37 AI054165, UM1 AI068635). The content is solely the responsibility of the authors and does not necessarily represent the official views of the National Institutes of Health. Statistical methodology development was supported by the National Science Foundation Graduate Research Fellowship Program (Grant No. DGE-2140004). Any opinions, findings, conclusions, or recommendations expressed in this material do not necessarily reflect the views of the National Science Foundation. We also would like to acknowledge the NIH R01AI192632 awarded to BZ and HJ.

\bibliographystyle{unsrtnat}
\bibliography{biblio}

\appendix

\section{Key proofs}

\subsection{Proposition 1}\label{prf:prop1}

As a reminder, Proposition 1 consists of two parts (a) applying G-computation with respect to observed versions of $N$ (e.g., $\wt{N}$ or $\Delta_N$) as suggested by \citep{Etievantetal2023} fails to identify $S_{a}(t_0)$ and (b) applying inverse probability of complete NCO case weighting along with G-computation successfully identifies $S_{a}(t_0)$. Below, we treat each of these results separately.

\subsubsection{Proposition 1(a)}

First, we show that identification of $S_{a_0}(t)$ fails when using a G-formula which adjusts for observed versions of $N$ such as $\wt{N}$ and $\Delta_N$. Without loss of generality, we can focus on adjusting for $\wt{N}$, knowing that identical arguments will follow for the observed event indicator $\Delta_N$.

Recall the definition of our target parameter written using the G-formula with respect to the full-data unit where $N$ is observed.
\begin{align*}
S_{a_0}(t)  = \mathbb{E}_{X,N}\left[S_{0Y}(t \mid A=a_0, X, N)\right]
\end{align*}
Consider the following estimator which marginalizes over the distribution of $X$ and the observed NCO event time $\wt{N}$.
\begin{align*}
\theta = \mathbb{E}_{X, \wt{N}}\left[P(Y > t \mid A=a_0, X, \wt{N})\right]
\end{align*}
By causal consistency and strong ignorability, we have that 
\begin{align*}
\mu_{a_0}(t,X,\wt{n}) =P(Y > t \mid A=a_0, X, \wt{N}=\wt{n}) = P(Y(a_0) > t \mid X, \wt{N}(a) = \wt{n})
\end{align*}
Where $\wt{N}(a) := \min(N, C_N(a))$. Hence,
\begin{align*}
    \theta \equiv \mathbb{E}_{X,\wt{N}}\left[\mu_{a_0}(t,X,\wt{N})\right]
\end{align*}
To illustrate why $\theta$ is biased for the target $S_{a_0}(t)$, we appeal to the ideas proposed by \citeay{Rosenbaum_1984} regarding adjustment for a post-treatment variable. In brief, the literature of principal stratification claims that the potential outcome under a specific arm, such as $\wt{N}(a)$, is an intrinsic, deterministic characteristic of a participant, meaning it behaves exactly like a baseline covariate.\citep{FrangakisRubin2002} Because $\wt{N}(a)$ is essentially a baseline covariate, standardizing with respect to its distribution must identify the marginal survival curve of interest.
\begin{align*}
    S_{a_0}(t)  = \mathbb{E}_{X,\wt{N}(a_0)}\left[P(Y(a) \geq t \mid A=a_0, X, \wt{N}(a_0))\right] \equiv  \mathbb{E}_{X,\wt{N}(a_0)}\left[\mu_{a_0}(t,X,\wt{N}(a_0))\right]
\end{align*}
Now we can see plainly why $\theta$ is biased for the target parameter. The failure of identification results from the improper marginalization of $\mu_{a_0}$ over the \textit{observed law} of $(X,\wt{N})$ instead of the \textit{counterfactual law} of $(X,\wt{N}(a_0))$. The observed law is a contaminated mixture of the two potential outcome distributions.
\begin{align*}
    P(\wt{N} \mid X) = \pi_0(a_0 \mid X) P(\wt{N}(a_0)\mid X) + \pi_0(1-a_0 \mid X) P(\wt{N}(1-a_0)\mid X)
\end{align*}

In general, $\mathbb{E}_{X, \wt{N}}[\mu_{a_0}(t, X, \wt{N})] \neq \mathbb{E}_{X, \wt{N}(a_0)}[\mu_{a_0}(t, X, \wt{N}(a_0)]$ when the distribution of $\wt{N} \overset{d}{\neq} \wt{N}(a_0)$ \citep{Rosenbaum_1984}. Hence, the failure of identification results from the fact that our target parameter marginalizes over the counterfactual distribution $\wt{N}(a_0)$ while the proposed parameter incorrectly marginalizes over the distribution of the observed variable $\wt{N}$. 

\subsubsection{Proposition 1(b)}

We establish identification based on inverse probability weighting (IPW) of complete cases. Our proposition is.
\begin{align*}
   S_{a_0}(t_0) = \mathbb{E}_0\left[\frac{I(N \leq C)}{G_0(\wt{N} \mid A, X)} \frac{I(A=a_0)}{\pi_0(a_0 \mid X)} P_0(Y > t \mid A=a_0, X, \wt{N}) \right]
\end{align*}
Where $\mathbb{E}_0[\cdot]$ refers to marginalizing over the distribution of $(A, X, N, C)$. 
\begin{align*}
    &\mathbb{E}_0\left[\frac{I(N \leq C)}{G_0(\wt{N} \mid A, X)} \frac{I(A=a_0)}{\pi_0(a_0 \mid X)} P_0(Y > t \mid A=a_0, X, \wt{N}) \right] \\
    &= \mathbb{E}_0\left[\frac{I(N \leq C)}{G_0(N \mid A = a_0, X)} \frac{I(A=a_0)}{\pi_0(a_0 \mid X)} P_0(Y > t \mid A=a_0, X, N) \right] \\
    &=\mathbb{E}_{X,N}\left[P_A(A=a_0 \mid X, N) \times \mathbb{E}_C\left\{\frac{I(N \leq C)}{G_{0}(N \mid A=a_0, X) \pi_0(a_0 \mid X)} P_0(Y > t \mid A=a_0, X,  N) \mid X, N, A=a_0\right\} \right] \\
    &=\mathbb{E}_{X,N}\left[\mathbb{E}_C\left\{\frac{I(N \leq C)}{G_{0}(N \mid A=a_0, X)} P_0(Y > t \mid A=a_0, X,  N) \mid X, N, A=a_0\right\} \right] \\
    &=\mathbb{E}_{X,N}\left[P_0(Y > t \mid A=a_0, X,  N) \cancel{\left\{\frac{G_{0}(N \mid A=a_0, X)}{G_{0}(N \mid A=a_0, X)}\right\}} \right] \\
    &= \mathbb{E}_{X,N}\left[S_{0Y}(t \mid a_0, X, N) \right] = S_{a_0}(t)
\end{align*}
The first equality follows from $\wt{N}=N$ when $\Delta_N=1$ with probability 1. The second equality follows from Tower Law and law of total expectation. The third equality follows from cancelling the treatment probabilities. The fourth follows from evaluating the expectation of the indicator function in the numerator and cancelling like terms. The final equality follows from expressing the survival function $S_a(t)$ in G-computation form and using the arguments developed in \citeay{Westling_2023}.

\subsection{Exclusion restriction \ref{rem:indep_1}}\label{prf:rem1}

In the remark below, we discuss two important consequences of independent censoring (Assumption 4).

The first exclusion restriction follows directly from $(C_Y, C_N) \perp (Y,N) \mid (A, X)$. To show the latter exclusion restriction, we break it into two cases.
\begin{align*}
\{C_N \perp N \mid (A,X,\Delta_Y, \wt{Y})\} &= \begin{cases}
        C_N \perp N \mid (A, X, Y, C_Y > Y) & \text{ when } \Delta_Y=1 \\
        C_N \perp N \mid (A, X, C_Y, Y \geq C_Y) & \text{ when } \Delta_Y=0
\end{cases}
\end{align*}
We examine each of these cases separately to verify that Assumption 4 implies the latter result.

When $\Delta_Y=1$, the conditional probability of $C_N$ and $N$ is given by
\begin{align*}
    &P(C_N, N \mid A, X, Y=y, C_Y \geq y) = \frac{P(C_N, N, C_Y \geq y, Y=y \mid A,X)}{P(C_Y \geq y, Y=y \mid A,X)} \\
    &=\frac{P(N, Y=y \mid A,X)}{P(Y=y \mid A,X)} \cdot \frac{P(C_N, C_Y \geq y \mid A,X)}{P(C_Y \geq y \mid A, X)} %= P(N \mid A,X,Y=y) \cdot P(C_N \mid A, X, C_Y \geq y)
\end{align*}
Where the first equality follows from Bayes rule and the second from applying Assumption 4 to the numerator and denominator. Since the joint probability factors into a term depending on $N$ and another depending on $C_N$, proving $C \perp N \mid A, X, \wt{Y}, \Delta_Y=1$.

A very similar result can be shown for the case when $\Delta_Y=0$:
\begin{align*}
    &P(C_N, N \mid A, X, C_Y=c_Y, Y>c_Y) = \frac{P(C_N, N, C_Y = y, Y>c_Y \mid A,X)}{P(C_Y = c_Y, Y > c_Y \mid A,X)} \\
    &=\frac{P(N, Y>c_Y \mid A,X)}{P(Y > c_Y \mid A,X)} \cdot \frac{P(C_N, C_Y = c_Y \mid A,X)}{P(C_Y = c_Y \mid A, X)}
\end{align*}
Where as before, Bayes rule and Assumption 4 show that the conditional likelihood can be factored into a term depending on $N$ and another depending on $C_N$, which proves $C \perp N \mid A, X, \wt{Y}, \Delta_Y=0$.

Taken together, we see Assumption 4 implies $C_N \perp N \mid (A,X,\wt{Y},\Delta_Y)$.

\subsection{Theorem 1} \label{prf:Thm1}

Recall that our target parameter is the treatment-arm-specific survivor function of $Y$ at some landmark time $t_0$, $S_{a_0}(t_0)$. As a reminder, we refer to the ``full" data as observations where the primary outcome $Y$ is right-censored (by $C_Y$) by the predictive, negative control infection time $N$ is observed for all units.
\[
O^{\text{full}} := (A, X, \wt{Y}, \Delta_Y, N)
\]
In practice, we only sample the ``observed" data, for which $N$ is also censored (by $C_N$).
\[
O^{\text{obs}} := (A, X, \wt{Y}, \Delta_Y, \wt{N}, \Delta_N)
\]
How shall we estimate $S_{a_0}(t_0)$ efficiently using the observed data? We will base our estimate on an estimating equation built using the observed-data efficient influence (EIF), commonly referred to as a one-step estimator \citep{Bickel_semi_1993}. 

To establish the observed-data EIF, we must build on the theory of full data influence functions. Let $\mathcal{H}^F$ denote the full-data Hilbert space, which consists of mean-zero measurable functions of $O^{\text{full}}$ with finite variance equipped with covariance inner product. Let $\mathcal{H}$ denote the observed-data Hilbert space, which consists of mean-zero measurable functions of $O^{\text{obs}}$ with finite variance equipped with covariance inner product. We focus on these Hilbert spaces because they represent plausible choices of estimating functions that can be used to construct estimating equations for $S_{a_0}(t_0)$. Note that influence functions of $S_{a_0}(t_0)$ in our full-data model lie in a subspace of $\mathcal{H}$ orthogonal to the nuisance tangent space. Hence, influence functions represent choices of estimating functions that are minimally influenced by perturbations of the nuisance parameters, which make them ideal choices for estimation.

A direct application of Theorem 1 of \citeay{Westling_2023} yields the EIF in the full-data model where $N$ is observed for all participants.
\begin{center}
\resizebox{\textwidth}{!}{
$\begin{aligned}
    \phi^F_{a_0,t_0}(o^{\text{full}}) &= S_{0Y}(t_0 \mid a_0, x, n) \left[1 - \frac{I(a=a_0)}{\pi_0(a_0 \mid x)} \left\{\frac{I(\wt{y} \leq t_0, \delta_y=1)}{S_{0Y}(\wt{y} \mid a, x, n) G_{0Y}(\wt{y} \mid a, x)} - \int_0^{\min(t_0, \wt{y})} \frac{\Lambda_{0Y}(du \mid a, x, n)}{S_{0Y}(u \mid a, x, n) G_{0Y}(u \mid a, x)}\right\}\right]
\end{aligned}$}
\end{center}
Note that $\pi_0(a_0 \mid x)$ does not depend on $N$ by Assumption 5, and $G_{0Y}(u \mid a, x)$ does not depend on $N$ by Assumptions 4 and 5. However, $S_{0Y}$ does depend on $N$, and the additional regressor is what drives the precision gain.

How do we proceed from the full-data model to the observed-data model? We make this transition by recognizing that the censoring of $N$ defines a coarsening of the full data unit, $O^{\text{full}}$. Following \citeay{Tsiatis2006-cc}, we represent the observed data as $O^{\text{obs}} = \{\mathcal{C}_N, G_{\mathcal{C}_N}(O^{\text{full}})\}$ where the coarsening level $\mathcal{C}_N$ is
\begin{align*}
    \mathcal{C}_N = \begin{cases}
        C_N & \text{if } C_N < N \\
        \infty & \text{if } C_N \geq N
    \end{cases}
\end{align*}
and $G_r(O^{\text{full}})$ is the corresponding many-to-one function of the full data satisfying
\begin{align*}
    G_r(O^{\text{full}}) = \begin{cases}
        (A,X,\wt{Y},\Delta_Y, C_N=r, N>r) & \text{if } r < \infty \\
        (A,X,\wt{Y},\Delta_Y, C_N \geq N, N) & \text{if } r = \infty
    \end{cases}
\end{align*}
Following \citeay{Tsiatis2006-cc}, we can express the coarsening variable's hazard at level $r$ as
\begin{align*}
    \lambda_r(O^{\text{full}}) = P(\mathcal{C}_N=r \mid \mathcal{C}_N \geq r, O^{\text{full}}).
\end{align*}
Using the equivalence $\{\mathcal{C}_N = r\} \equiv (C_N \geq r, N > r) \cup (N \leq C_N)$; the events overlap on $\{r < N \leq \mathcal{C}_n\}$. We can rewrite the hazard in terms of the censoring variable $C_N$.
\begin{align*}
    \lambda_r(O^{\text{full}}) = P(C_N=r, N>r \mid (C_N \geq r, N > r) \cup (N \leq C_N), O^{\text{full}})
\end{align*}
On the event $\{N<r\}$, the hazard is zero, so we can restrict attention to the event $\{N \geq r\}$, on which the union in the conditioning event reduces to $\{C_N \geq r\}$.
\begin{align*}
    \lambda_r(O^{\text{full}}) = I(N \geq r) P(C_N=r \mid C_N \geq r, O^{\text{full}}) = I(N \geq r) \lambda_{C_N}(r \mid O^{\text{full}}) 
\end{align*}
Where $\lambda_{C_N}(r \mid O^{\text{full}}$ is the hazard of $C_N$ at level $r$ conditional on the full data unit. Under coarsening at random (i.e., independent censoring), we have that the $\lambda_{C_N}(r \mid O^{\text{full}}) = \lambda_{C_N}(r \mid G_r(O^{\text{full}}))$, yielding
\begin{align*}
    \lambda_r(O^{\text{full}}) &= I(N \geq r) \lambda_{C_N}(r \mid A, X, \wt{Y}, \Delta_Y)
\end{align*}
Where $\lambda_{C_N}(r \mid A, X, \wt{Y}, \Delta_Y)$ is the hazard of $C_N$ at level $r$ conditional on $(A,X,\wt{Y},\Delta_Y)$.

Define the conditional survival and cumulative hazard functions of $C_N$ given $(A,X,\wt{Y},\Delta_Y)$ as follows:
\begin{align*}
    G_{0N}(t \mid a,x,\wt{y},\delta_y) := \exp\left\{\int_0^t \lambda_{C_N}(u \mid a,x,\wt{y},\delta_y) du \right\} \qquad \Lambda_{C_N}(t \mid a, x, \wt{y}, \delta_y) := \int_0^t \lambda_{C_N}(u \mid a,x,\wt{y},\delta_y) du
\end{align*}
Let $N_C(t) := I(\wt{N} \leq t, \Delta_N=0\}$ denote the counting process for censoring and let $Y_R(t) := I(\wt{N} \geq t)$ denote the at risk indicator. Under \ref{rem:indep_1} implied by Assumption 4, the process
\begin{align*}
    M_{C_N}(t) := N_C(t) - \int_0^t Y_R(u) d\Lambda_{C_N}(u \mid a, x, \wt{y}, \delta_y)
\end{align*}
is a martingale with respect to the filtration generated by the observed data on the time scale of $C_N$.

By Theorem 10.4 of \citeay{Tsiatis2006-cc} (in its continuous-time form, equation 10.76) and using the fact that $\phi^F_{a_0, t_0}$ has finite variance, the projection of $\phi^F_{a_0, t_0}$ onto the observed-data Hilbert space is as follows.
\begin{align*}
    \phi_{a_0, t_0}^{\text{obs}} &= \frac{\Delta_N \phi^F_{a_0, t_0}(\wt{Y}, \Delta_Y, A, X, \wt{N})}{G_{0N}(\wt{N} \mid A, X, \wt{Y}, \Delta_Y)} + \int_0^{\infty} \frac{h^*(u, A, X, \wt{Y}, \Delta_Y)}{G_{0N}(u \mid A, X, \wt{Y}, \Delta_Y)} dM_{C_N}(u \mid A, X, \wt{Y}, \Delta_Y)
\end{align*}
Where $h^*$ is an augmentation function chosen to make $\phi_{a_0,t_0}^{\text{obs}}$ orthogonal to the nuisance tangent space contributed by the $C_{N}$ coarsening mechanism. The optimal choice (in the sense of the observed data influence function with lowest variance) is
\begin{align*}
    h^*(u, a, x, \wt{y}, \delta_y) = Q_0(u \mid a, x, \wt{y}, \delta_y) := \mathbb{E}[\phi^F_{a_0, t_0}(\wt{Y}, \Delta_Y, A, X, N) \mid N \geq u, A=a, X=x, \wt{Y}=\wt{y}, \Delta_Y = \delta_y]
\end{align*}
the conditional expectation of the full-data influence function over the at-risk set $\{N \geq u\}$, given baseline factors and primary outcome history. 

Substituting $Q_0$ in for $h^*$ and expanding the censoring martingale into two terms, $dM_{C_N}(u) = dN_{C}(u) - Y_R(u) d\Lambda_{C_N}(u)$, the integral term can be writen as
\begin{align*}
    \int_0^{\infty} \frac{Q_0(u)}{G_{0N}(u)} dM_{C_N}(u) = \frac{(1-\Delta_N) Q_0(\wt{N})}{G_{0N}(\wt{N} \mid A, X, \wt{Y}, \Delta_Y)} - \int_0^{\wt{N}} \frac{Q_{0}(u \mid A, X, \wt{Y}, \Delta_Y)}{G_{0N}(u \mid A, X, \wt{Y}, \Delta_Y)} d\Lambda_{C_N}(u \mid A, X, \wt{Y}, \Delta_Y)
\end{align*}
Where we express some conditioning for the sake of compactness. Combining with the result from Theorem 10.4 above, 
\begin{align*}
    \phi_{a_0, t_0}^{\text{obs}} &= \frac{\Delta_N \phi^F_{a_0, t_0}(\wt{Y}, \Delta_Y, A, X, \wt{N})}{G_{0N}(\wt{N} \mid A, X, \wt{Y}, \Delta_Y)} + \frac{(1-\Delta_N) Q_0(\wt{N} \mid A,X,\wt{Y},\Delta_Y)}{G_{0N}(\wt{N} \mid A, X, \wt{Y}, \Delta_Y)} \\
    &- \int_0^{\wt{N}} \frac{Q_{0}(u \mid A, X, \wt{Y}, \Delta_Y)}{G_{0N}(u \mid A, X, \wt{Y}, \Delta_Y)} d\Lambda_{C_N}(u \mid A, X, \wt{Y}, \Delta_Y)
\end{align*}
The expression has the form of an augmented inverse probability of censoring complete-case (AIPWCC) estimating function, matching the influence function showed in \ref{eq:AIPWCC_IF}: the first term is the inverse-probability-weighted complete-case contribution from participants whose $N$ is observed; the second and third terms together form the augmentation that recovers efficiency by exploiting partial information from $N$-censored participants. 

Now we wish to show that the influence function $\phi^{\text{obs}}_{a_0, t_0}$ is the efficient influence function globally. Note that Theorem 10.4 supports that $\phi^{\text{obs}}_{a_0, t_0}$ is the ``optimal" choice of observed-data influence function based on a particular choice of full-data influence function $\phi^F_{a_0, t_0}(\cdot)$. However, in general, there is no guarantee that our choice of $\phi^F_{a_0, t_0}(\cdot)$ leads to the optimal observed data influence function.

However, in this case, we can rightly claim that $\phi^{\text{obs}}_{a_0, t_0}$ is globally efficient. By Theorem 1 of \citeay{Westling_2023}, $\phi^F_{t_0, a_0}$ is the \textit{unique/nonparametric} EIF in the full-data model. Since the full-data influence function is unique, the projection onto the observed data Hilbert space is also unique. Hence, $\phi^{\text{obs}}_{a_0, t_0}$ must be the globally efficient influence function in the observed-data model.\qed

\subsection{Proposition 2}\label{prf:prop2}

Recall our goal is to derive the relationship between the remaining nuisance parameter in the EIF $\phi^{\text{obs}}_{a_0, t_0}$ -- $f_{0N}(n \mid a, x, \wt{y}, \delta_y)$ -- and the nuisances parameters compatible with a ``one-way" factorization of the joint likelihood $(Y,N) \mid (A,X)$ -- $f_{N}(n \mid a, x)$ and $S_{0Y}(y \mid a, x, n)$. The key piece of our derivation is an application of Bayes rule.

Specifically, we decompose the target nuisance parameter according to the event indicator $\delta_y$.
\begin{align*}
    f_{0N}(n \mid a, x, \wt{y}, \delta_y) &= \delta_y f(n \mid a, x, c_y>\wt{y}, y=\wt{y}) + (1-\delta_y) f_{N}(n \mid a, x, c_Y=\wt{y}, y > \wt{y}) \\
    &= \delta_y f(n \mid a, x, y=\wt{y}) + (1-\delta_y) f(n \mid a, x, y >\wt{y})
\end{align*}
The cancellation of the terms $c_y > \wt{y}$ and $c_y = \wt{y}$ in the second line is the result of independent censoring (Assumption 4). Next, we apply Bayes rule to the constituent pieces.
\begin{align*}
    f_{0N}(n \mid a, x, \wt{y}, \delta_y) &= \delta_y \left(\frac{f_{0Y}(Y=\wt{y} \mid a, x, n) f(n \mid a,x)}{f_Y(Y=\wt{y} \mid a, x)}\right) + (1-\delta_y) \frac{f(n \mid a, x)S_{0Y} (\wt{y} \mid a, x, n)}{S_Y(\wt{y} \mid a, x)} \\
    &= f_N(n \mid a, x) \cdot \frac{L(\wt{y}, \delta_y \mid a, x, n)}{L(\wt{y}, \delta_y \mid a, x)}
\end{align*}
where
\begin{align*}
    \frac{L(\wt{y}, \delta_y \mid a, x, n)}{L(\wt{y}, \delta_y \mid a, x)} &=\left[\delta_y \left(\frac{f_{0Y}(\wt{y} \mid a, x, n)}{f_{Y}(\wt{y} \mid a, x)}\right) + (1-\delta_y) \frac{S_ {0Y}(\wt{y} \mid a, x, n)}{S_Y(\wt{y} \mid a, x)}\right].
\end{align*}
$S_{Y}(u \mid a,x) = \int_0^{\infty} S_{0Y}(u \mid a, x, n) f_{N}(n \mid a, x) dn$ by Tower law. $f_{0Y}$ and $f_{Y}$ refer to the conditional density functions associated with $S_{0Y}$ and $S_Y$, which can be obtained from the survival functions using the following identity $f(x) = - (d/du) S(u) \mid_{u=x}$.

Hence, a compatible estimator of $f_{0N}(n \mid a, x, \wt{y}, \delta_y)$ can be obtained from the nuisances that define a one-way factorization of the joint likelihood: $f_{N}(n \mid a, x)$ and $S_{0Y}(y \mid a, x, n)$. \qed

\subsection{Theorem \ref{th:MR}}\label{prf:2}

\begin{Lem}\label{lem:MR}
   Let $\eta=(\pi,G_Y,G_N,S_Y,f_N)$ be the nuisance parameter vector, with true value $
\eta_0=(\pi_0,G_{0Y},G_{0N},S_{0Y},f_{0N})$.
Suppose Assumptions 1--6 hold,
%and standard regularity conditions hold, %including the conditional independent censoring %conditions for $Y$ and $N$. 
then
\bse
E_0\{\phi_{a_0,t_0}^{\mathrm{obs}}(O^{\mathrm{obs}};\eta)\}=S_{a_0}(t_0)
\ese
if any one of the following nuisance-function sets is correctly specified:
\bse
	&&(S_Y, G_N) = (S_{0Y}, G_{0N}), \qquad
	(S_Y, f_N) = (S_{0Y}, f_{0N}), \\
	&&(\pi, G_Y, G_N) = (\pi_0, G_{0Y},G_{0N}),\qquad
	(\pi, G_Y, f_N) = (\pi_0, G_{0Y},f_{0N}).
\ese
\end{Lem}
\begin{proof}[Proof of Lemma \ref{lem:MR}]
	For a generic nuisance vector $\eta$, define
    \be\label{eq:Q}
    Q(u\mid a,x,\wt y,\delta_y)
    \equiv\E\left\{
    \phi^F_{a_0,t_0}(\wt y,\delta_y,a,X,N;\eta)
    \mid N\geq u,A=a,X=x,\wt Y=\wt y,\Delta_Y=\delta_y
    \right\}.
    \ee
    For notational simplicity, we write this as $Q(u)$ whenever the conditioning variables and nuisance value are clear from context.
    
    Write
	\be\label{eq:phi_decompose}
	\phi^{\mathrm{obs}}_{a_0,t_0}=T_1+T_2-T_3,
	\ee
	where
	\bse
	T_1&=&\frac{\Delta_N}{G_N(\wt N\mid a_0,X,\wt Y,\Delta_Y)}
	\phi^F_{a_0,t_0}(\wt Y,\Delta_Y,a_0,X,\wt N),\\
	T_2&=&\frac{1-\Delta_N}{G_N(\wt N\mid a_0,X,\wt Y,\Delta_Y)}
	Q(\wt N\mid \wt a_0,X,\wt Y,\Delta_Y),\\
	T_3&=&\int_0^{\wt N}\frac{Q(u\mid a_0,X,\wt Y,\Delta_Y)}{G_N(u\mid a_0,X,\wt Y,\Delta_Y)}d\Lambda_{C_N}(u\mid a_0,X,\wt Y,\Delta_Y).
	\ese
    
	We prove the result case by case.
	
	\textbf{Case 1: $(S_Y,G_N)=(S_{0Y},G_{0N})$.}
	Under this case, let
	\bse
	\phi_{a_0,t_0}^F(o^{\text{full}};\eta)= S_{0Y}(t_0 \mid a_0, x, n) \left\{1 - \frac{I(a=a_0)}{\pi(a_0 \mid x)} H_{S_{0Y},G_Y,t_0,a_0}(\wt{y},\delta_y,x,n)\right\},
	\ese
	where
	\be\label{eq:H}
	H_{S_Y,G_Y,t_0,a_0}(\wt{y},\delta_y,x,n)=\frac{I(\wt{y} \leq t_0, \delta_y=1)}{S_Y(\wt{y} \mid a_0, x, n) G_Y(\wt{y} \mid a_0, x)} - \int_0^{\min(t_0, \wt{y})} \frac{\Lambda_Y(du \mid a_0, x, n)}{S_Y(u \mid a_0, x, n) G_Y(u \mid a_0, x)}.
	\ee
	We note that $\E\{H_{S_{0Y},G_Y,t_0,a_0}(\wt{Y},\Delta_Y,X,N)\mid A=a_0,X=x,N=n\}$ equals
	\bse
	\int_0^{t_0} \frac{S_{0Y}(y-\mid a_0,x,n)G_{0Y}(y\mid a_0,x)}{S_{0Y}(y\mid a_0,x,n)G_Y(y\mid a_0,x)}\Lambda_{0Y}(dy \mid a_0,x,n)-\int_0^{t_0} \frac{S_{0Y}(u-\mid a_0,x,n)G_{0Y}(u\mid a_0,x)}{S_{0Y}(u \mid a_0, x, n) G_Y(u \mid a_0, x)}\Lambda_{0Y}(du \mid a_0, x, n)=0,
	\ese
	then the full-data EIF is unbiased:
	\bse
	\E\{\phi^F_{a_0,t_0}(O^{\mathrm{full}};\eta)\}=\E\{S_{0Y}(t_0\mid a_0,X,N)\}=S_{a_0}(t_0) .
	\ese

	For the first term $T_1$, under Case 1, we first have
	\bse
	\E \left\{ \frac{\Delta_N}{G_{0N}(\wt{N}\mid a_0,X,\wt{Y},\Delta_Y)}\phi_{a_0, t_0}^F \mid \wt{Y},\Delta_Y,a_0, X,N\right\}&=&\E \left\{ \frac{I(N\le C_N)}{G_{0N}(\wt{N}\mid a_0,X,\wt{Y},\Delta_Y)}\phi_{a_0, t_0}^F \mid \wt{Y},\Delta_Y,a_0, X,N\right\}\\
	&=&\phi_{a_0, t_0}^F(\wt{Y},\Delta_Y,a_0, X,N)\cdot \frac{\E(C_N\ge N\mid \wt{Y},\Delta_Y,a_0, X,N)}{G_{0N}(\wt{N}\mid a_0,X,\wt{Y},\Delta_Y)}\\
	&=&\phi_{a_0, t_0}^F(\wt{Y},\Delta_Y,a_0, X,N).
	\ese
	Thus, $\E(T_1)=\E\{\Delta_N\phi_{a_0,t_0}^F/G_{0N}(\wt{N}\mid a_0,X,\wt{Y},\Delta_Y)\}=S_{a_0}(t_0)$.
    
	For the second term $T_2$, under Case 1, its expectation is
	\bse
	\E\left\{ \frac{1-\Delta_N}{G_{0N}(\wt{N}\mid a_0,X,\wt{Y},\Delta_Y)}Q(\wt{N})  \right\}&=&\E\left\{ \frac{I(C_N<N)}{G_{0N}(C_N\mid a_0,X,\wt{Y},\Delta_Y)}Q(C_N)  \right\}\\
	&=& \E\left\{ \int_0^N \frac{Q(u)}{G_{0N}(u\mid a_0,X,\wt{Y},\Delta_Y)}G_{0N}(u\mid a_0,X,\wt{Y},\Delta_Y)d\Lambda_{0C_N}(u\mid a_0,X, \wt{Y},\Delta_Y)\right\}\\
	&=& \E\left\{\int_0^N Q(u)d\Lambda_{0C_N}(u\mid a_0,X, \wt{Y},\Delta_Y)\right\}.
	\ese 
	The expectation of the third term $T_3$ equals to
	\bse
%	E_0\left\{ \int_0^\infty I(\wt{N}\ge u)\frac{Q^*(u)}{G_0(u\mid a_0,X)} d\Lambda_C(u\mid a_0,X)   \right\}&=&
	&&\E\left\{ \int_0^\infty I(N\ge u)I(C_N\ge u)\frac{Q(u)}{G_{0N}(u\mid a_0,X,\wt{Y},\Delta_Y)} d\Lambda_{0C_N}(u\mid a_0,X, \wt{Y},\Delta_Y)   \right\}\\
	&=& \E\left[\int_0^\infty \E\left\{I(N\ge u)I(C_N\ge u)\mid \wt{Y},\Delta_Y,a_0,X,N\right\}\frac{Q(u)}{G_{0N}(u\mid a_0,X,\wt{Y},\Delta_Y)}d\Lambda_{0C_N}(u\mid a_0,X, \wt{Y},\Delta_Y)\right]\\
	&=&\E\left\{\int_0^N Q(u)d\Lambda_{0C_N}(u\mid a_0,X, \wt{Y},\Delta_Y)\right\}
	\ese
    Therefore, $T_2$ and $T_3$ cancel in expectation, so that
    \bse
    E_0\{\phi^{\mathrm{obs}}_{a_0,t_0}(O^{\mathrm{obs}};\eta)\}
    =E(T_1)=S_{a_0}(t_0),
    \ese
    which establishes Case 1.
    
	\textbf{Case 2: $(S_Y,f_N)=(S_{0Y},f_{0N})$.}
	Again, $S_Y^*=S_{0Y}$ implies that the full-data EIF is
	unbiased, as established in Case 1:
	\bse
	\E\{\phi^F_{a_0,t_0}(O^{\mathrm{full}};\eta)\}=S_{a_0}(t_0) .
	\ese
	The expectation of $T_1$ becomes
	\bse
	\E\left[\E\left\{\frac{\Delta_N}{G_N(\wt{N}\mid a_0,X,\wt{Y},\Delta_Y)}\phi_{a_0,t_0}^F \mid \wt{Y},\Delta_Y,a_0,X,N\right\}\right]&=&\E\left[\E\left\{\frac{I(N\le C_N)}{G_N(N\mid a_0,X,\wt{Y},\Delta_Y)}\phi_{a_0,t_0}^F \mid \wt{Y},\Delta_Y,a_0,X,N\right\}\right]\\
	&=&\E\left\{  \frac{G_{0N}(N\mid a_0,X,\wt{Y},\Delta_Y)}{G_N(\wt{N}\mid a_0,X,\wt{Y},\Delta_Y)}\phi_{a_0,t_0}^F\right\}.
	\ese
	The expectation of $T_2$ becomes
	\bse
	\E\left\{ \int_0^N \frac{Q(u)}{G_N(u\mid a_0,X,\wt{Y},\Delta_Y)}G_{0N}(u\mid a_0,X,\wt{Y},\Delta_Y)d\Lambda_{0C_N}(u\mid a_0,X,\wt{Y},\Delta_Y)\right\},
	\ese
	while the expectation of $T_3$ becomes
	\bse
	\E\left\{ \int_0^N \frac{Q(u)}{G_N(u\mid a_0,X,\wt{Y},\Delta_Y)}G_{0N}(u\mid a_0,X,\wt{Y},\Delta_Y)d\Lambda_{C_N}(u\mid a_0,X,\wt{Y},\Delta_Y)\right\}.
	\ese
	Define $W(u\mid a,x,\wt{y},\delta_y)\equiv G_{0N}(u\mid a,x,\wt{y},\delta_y)/G_N(u\mid a, x,\wt{y},\delta_y)$, we have
	\bse
	dW(u\mid a_0,x,\wt{y},\delta_y)=W(u\mid a_0,x,\wt{y},\delta_y)\{d\Lambda_{C_N}(u\mid a_0,x,\wt{y},\delta_y)-d\Lambda_{0C_N}(u\mid a_0,x,\wt{y},\delta_y)\}.
	\ese
	Thus, 
	\bse
    \E(T_2)-\E(T_3)&=&
	\E\left\{ \int_0^N \frac{Q(u)}{G_N(u\mid a_0,X,\wt{Y},\Delta_Y)}G_{0N}(u\mid a_0,X,\wt{Y},\Delta_Y)d\Lambda_{0C_N}(u\mid a_0,X,\wt{Y},\Delta_Y)\right\}\\
	&&-\E\left\{ \int_0^N \frac{Q_(u)}{G_N(u\mid a_0,X,\wt{Y},\Delta_Y)}G_{0N}(u\mid a_0,X,\wt{Y},\Delta_Y)d\Lambda_{C_N}(u\mid a_0,X,\wt{Y},\Delta_Y)\right\}\\
	&=& -\E\left\{\int_0^N Q(u)dW(u\mid a_0,X,\wt{Y},\Delta_Y)   \right\}.
	\ese
	Further, we have
	\bse
	\E\left\{\int_0^N Q(u)dW(u\mid a_0,X,\wt{Y},\Delta_Y)\mid \wt{Y},\Delta_Y,a_0,X\right\}
	&=&
	\int_0^\infty
	\E\left\{I(N\ge u)Q(u)\mid \wt{Y},\Delta_Y,a_0,X\right\}
	dW(u\mid a_0,X,\wt{Y},\Delta_Y)\\
	&=&
	\int_0^\infty
	\E\left\{I(N\ge u)\phi^F_{a_0,t_0}\mid \wt{Y},\Delta_Y,a_0,X\right\}
	dW(u\mid a_0,X,\wt{Y},\Delta_Y).
	\ese
	In addition, from $W(N\mid a_0,X,\wt{Y},\Delta_Y)=1+\int_0^\infty I(N\ge u)dW(u\mid a_0,X,\wt{Y},\Delta_Y)$, we have
	\bse
	&&\E\left\{W(N\mid a_0,X,\wt{Y},\Delta_Y)\phi^F_{a_0,t_0}\mid \wt{Y},\Delta_Y,a_0,X\right\}\\
	&=&
	\E\{\phi^F_{a_0,t_0}\mid \wt{Y},\Delta_Y,a_0,X\}+
	\int_0^\infty
	\E\left\{I(N\ge u)\phi^F_{a_0,t_0}\mid \wt{Y},\Delta_Y,a_0,X\right\}
	dW(u\mid a_0,X,\wt{Y},\Delta_Y).
	\ese
	Combining the above two displays, conditional on $(\wt Y,\Delta_Y,a_0,X)$, we have
	\bse
    E(T_1+T_2-T_3\mid \wt Y,\Delta_Y,a_0,X)&=&	\E\left[\left\{
	W(N\mid a_0,X,\wt{Y},\Delta_Y)\phi^F_{a_0,t_0}
	-
	\int_0^N Q(u)dW(u\mid a_0,X,\wt{Y},\Delta_Y)\right\}
	\mid \wt Y,\Delta_Y,a_0,X
	\right] \\
	&=&
	\E\{\phi^F_{a_0,t_0}\mid \wt Y,\Delta_Y,a_0,X\}.
	\ese
	Taking expectation on both sides yields
	\bse
	\E\left\{\phi^{\mathrm{obs}}_{a_0,t_0}(O^{\text{obs}};\eta)\right\}
	&=&
	\E\left\{
	W(N\mid a_0,X,\wt{Y},\Delta_Y)\phi^F_{a_0,t_0}
	-
	\int_0^N Q(u)dW(u\mid a_0,X,\wt{Y},\Delta_Y)
	\right\}\\
	&=&
	\E\{\phi^F_{a_0,t_0}(O^{\mathrm{full}};\eta)\}=S_{a_0}(t_0),
	\ese
    which establishes Case 2.
	
	\textbf{Case 3: $(\pi,G_Y,G_N)=(\pi_0,G_{0Y},G_{0N})$.} Under this case, we have
	\bse
	\phi_{a_0,t_0}^F(\wt{y}, \delta_y, x,n;\eta)= S_Y(t_0 \mid a_0, x, n) \left\{1 - \frac{I(a=a_0)}{\pi_0(a_0 \mid x)} H_{S_Y,G_{0Y},t_0,a_0}(\wt{y},\delta_y,x,n)\right\},	
	\ese
	where $H_{S_Y,G_Y,t_0,a_0}(\wt{y},\delta_y,x,n)$ is defined in \eqref{eq:H}. 
	
	We note that $\E\{H_{S_Y,G_{0Y},t_0,a_0}(\wt{Y},\Delta_Y,X,N)\mid A=a_0,X=x,N=n\}$ equals to
	\bse
	&&\int_0^{t_0} \frac{S_{0Y}(y-\mid a_0,x,n)G_{0Y}(y\mid a_0,x)}{S_Y(y\mid a_0,x,n)G_{0Y}(y\mid a_0,x)}\Lambda_{0Y}(dy \mid a_0,x,n)-\int_0^{t_0} \frac{S_{0Y}(u-\mid a_0,x,n)G_{0Y}(u\mid a_0,x)}{S_Y(u \mid a_0, x, n) G_{0Y}(u \mid a_0, x)}\Lambda_Y(du \mid a_0, x, n)\\
	&=& \int_0^{t_0}\frac{S_{0Y}(u-\mid a_0,x,n)}{S_Y(u\mid a_0,x,n)}(\Lambda_{0Y}-\Lambda_Y)(du\mid a_0,x,n).
	\ese
	The Duhamel equation applied to $S_Y-S_{0Y}$ gives
	\bse
	S_Y(t\mid a_0,x,n)-S_{0Y}(t\mid a_0,x,n)=S_Y(t\mid a_0,x,n)\int_0^t \frac{S_{0Y}(u-\mid a_0,x,n)}{S_Y(u\mid a_0,x,n)}(\Lambda_{0Y}-\Lambda_Y)(du\mid a_0,x,n).
	\ese
	Thus, $\E\{H_{S_Y,G_{0Y},t_0,a_0}(\wt{Y},\Delta_Y,X,N)\mid A=a_0,X=x,N=n\}$ equals to
	\bse
	\frac{S_Y(t_0\mid a_0,x,n)-S_{0Y}(t_0\mid a_0,x,n)}{S_Y(t_0\mid a_0,x,n)}.
	\ese
	Therefore, under Case 3, 
	\bse
	\E\left\{\phi_{a_0,t_0}^F(\wt{Y},\Delta_Y,X,N;\eta)\mid A=a_0,X=x,N=n\right\}
	&=&
	S_Y(t_0\mid a_0,x,n)
	\left[
	1-
	\frac{S_Y(t_0\mid a_0,x,n)-S_{0Y}(t_0\mid a_0,x,n)}
	{S_Y(t_0\mid a_0,x,n)}
	\right]\\
	&=&
	S_{0Y}(t_0\mid a_0,x,n).
	\ese
	Taking expectation over $(X,N)$ yields
	\bse
	\E\left\{\phi_{a_0,t_0}^F(O^{\text{full}};\eta)\right\}=S_{a_0}(t_0).
	\ese
	
	Since $G_N=G_{0N}$, we have proved in Case 1 that 
	\bse
	\E\{\phi^{\mathrm{obs}}_{a_0,t_0}(O^{\text{obs}};\eta)\}
	=
	\E\{\phi_{a_0,t_0}^F(O^{\mathrm{full}};\eta)\}=S_{a_0}(t_0),
	\ese
    which establishes Case 3.
	
	\textbf{Case 4: $(\pi,G_Y,f_N)=(\pi_0,G_{0Y},f_{0N})$.} We have proved in Case 3 that when $\pi=\pi_0$ and $G_Y=G_{0Y}$,
	\bse
	\E\{\phi_{a_0,t_0}^F(O^{\mathrm{full}};\eta)\}=S_{a_0}(t_0).
	\ese
	In addition, when $f_N=f_{0N}$, from Case 2, we have shown that
	\bse
	\E\{\phi_{a_0,t_0}^{\mathrm{obs}}(O^{\mathrm{obs}};\eta)\}=E\{\phi_{a_0,t_0}^F(O^{\mathrm{full}};\eta)\}=S_{a_0}(t_0),
	\ese
    which establishes Case 4.
	
	Combining the four cases, $\E\{\phi_{a_0,t_0}^{\mathrm{obs}}(O^{\mathrm{obs}};\eta)\}=S_{a_0}(t_0)$
	whenever one of
	\bse
	&&(S_Y, G_N) = (S_{0Y}, G_{0N}), \qquad
	(S_Y, f_N) = (S_{0Y}, f_{0N}), \\
	&&(\pi, G_Y, G_N) = (\pi_0, G_{0Y},G_{0N}),\qquad
	(\pi, G_Y, f_N) = (\pi_0, G_{0Y},f_{0N}).
	\ese
	holds.

\end{proof}

\begin{Lem}
	\label{lem:eif-consistency}
	Let
	\bse
    \phi^{\mathrm{obs}}_{\infty,a_0,t_0}=\phi^{\mathrm{obs}}_{a_0,t_0}(O^{\mathrm{obs}};\eta_\infty), \qquad \phi^{\mathrm{full}}_{\infty,a_0,t_0}=\phi^{\mathrm{full}}_{a_0,t_0}(O^{\mathrm{full}};\eta_\infty),
	\ese
	where
	\bse
	\eta_\infty=(\pi_\infty,G_{\infty,Y},G_{\infty,N},S_{\infty,Y},f_{\infty,N}).
	\ese
	Assume Conditions \ref{assump:c1}---\ref{assump:c2} hold, then
	\bse
	\max_k\left\|\wh\phi^{\mathrm{obs}}_{n,k,a_0,t_0}-\phi^{\mathrm{obs}}_{\infty,a_0,t_0}
	\right\|_2=o_p(1).
	\ese
\end{Lem}

\begin{proof}[Proof of Lemma \ref{lem:eif-consistency}]
	Write
	\bse
	\phi^{\mathrm{obs}}_{a_0,t_0}=T_1+T_2-T_3,
	\ese
	where
	\bse
	T_1&=&\frac{\Delta_N}{G_N(\wt N\mid a_0,X,\wt Y,\Delta_Y)}
	\phi^F_{a_0,t_0}(\wt Y,\Delta_Y,a_0,X,\wt N),\\
	T_2&=&\frac{1-\Delta_N}{G_N(\wt N\mid a_0,X,\wt Y,\Delta_Y)}
	Q(\wt N\mid a_0,X,\wt Y,\Delta_Y),\\
	T_3&=&\int_0^{\wt N}\frac{Q(u\mid a_0,X,\wt Y,\Delta_Y)}{G_N(u\mid a_0,X,\wt Y,\Delta_Y)}d\Lambda_{C_N}(u\mid a_0,X,\wt Y,\Delta_Y),
	\ese
    where $Q(u\mid a,x,\wt y,\delta_y)$ is defined in \eqref{eq:Q}. For simplicity, the conditioning variables
	$(a_0,X,\widetilde Y,\Delta_Y)$ will be suppressed in the notation.

    For each fold $k$, let $
    \wh T_{j,k}$, $j=1,2,3$,
    denote the corresponding terms of $T_j$ obtained by replacing the nuisance functions with their cross-fitted estimators
    $
    \wh\eta_{n,k}=(\widehat\pi_{n,k},\widehat G_{n,k,Y},\widehat G_{n,k,N},
    \widehat S_{n,k,Y},\widehat f_{n,k,N})
    $.
    Similarly, let $
    T_{j,\infty}$, $j=1,2,3$
    denote the same terms with the nuisance functions replaced by their probability limits
    $\eta_\infty$.
    
	To analyze the first term $T_1$, we have
	\bse
	\wh T_{1,k}-T_{1,\infty}
	=
	\frac{\Delta_N}{\wh G_{n,k,N}(\wt N)}\{\wh\phi^F_{n,k,a_0,t_0}-\phi^F_{\infty,a_0,t_0}\}
	+\Delta_N\phi^F_{\infty,a_0,t_0}\left\{\frac1{\wh G_{n,k,N}(\wt N)}
	-\frac1{G_{\infty,N}(\wt N)}\right\}.
	\ese
	By Condition \ref{assump:c2},
	\bse
	\left\|\frac{\Delta_N}{\wh G_{n,k,N}(\wt N)}\{\wh\phi^F_{n,k,a_0,t_0}-\phi^F_{\infty,a_0,t_0}\}
	\right\|_2=O_p(\|\wh\phi^F_{n,k,a_0,t_0}-\phi^F_{\infty,a_0,t_0}\|_2) .
	\ese
	By the same argument as Lemma 3 of \cite{Westling_2023}, together with Conditions \ref{assump:c1} and \ref{assump:c2}, we have
	\bse
	\max_k
	\|\wh\phi^F_{n,k,a_0,t_0}-\phi^F_{\infty,a_0,t_0}\|_2=o_p(1).
	\ese
    By Condition \ref{assump:c2}, we can see that the full-data EIF $\phi^F_{\infty,a_0,t_0}$ is bounded. Under the boundedness of the full-data EIF and Condition \ref{assump:c1}, we have
	\bse
	\left\|\Delta_N\phi^F_{\infty,a_0,t_0}\left\{\frac1{\wh G_{n,k,N}(\wt N)}
	-\frac1{G_{\infty,N}(\wt N)}\right\}\right\|_2=
	O_p\left(\left\|\frac1{\wh G_{n,k,N}}-\frac1{G_{\infty,N}}\right\|_2\right)
	=o_p(1).
	\ese
	Hence,
	\bse
	\max_k\|\wh T_{1,k}-T_{1,\infty}\|_2=o_p(1).
	\ese
    Define
    \bse
    \wh Q_{n,k}(u)&\equiv&
    Q\left(u\mid A,X,\wt Y,\Delta_Y;\wh\eta_{n,k}
    \right),\\
    Q_\infty(u)&\equiv&Q\left(
    u\mid A,X,\wt Y,\Delta_Y;\eta_\infty\right).
    \ese
    For notational simplicity,
    the conditioning variables $(A,X,\wt Y,\Delta_Y)$ are suppressed.

	To analyze the second term $T_2$, we first have
	\bse
	\wh T_{2,k}-T_{2,\infty}=\frac{1-\Delta_N}{\wh G_{n,k,N}(\wt N)}\left\{\wh Q_{n,k}(\wt N)-Q_\infty(\wt N)\right\}+(1-\Delta_N)Q_\infty(\wt N)\left\{\frac1{\wh G_{n,k,N}(\wt N)}
	-\frac1{G_{\infty,N}(\wt N)}\right\}.
	\ese
	By Condition \ref{assump:c1}, \ref{assump:c2} and boundedness of $Q_\infty$ (obtained by boundedness $\phi_{\infty,a_0,t_0}^F$),
	\bse
	\max_k\|\wh T_{2,k}-T_{2,\infty}\|_2&=&\max_k\ O_p(\|\wh Q_{n,k}(\wt N)-Q_\infty(\wt N)\|_2)+\max_k\ O_p\left(\left\|\frac1{\wh G_{n,k,N}}-\frac1{G_{\infty,N}}\right\|_2\right)\\
	&=&o_p(1).
	\ese
	
	For the third term $T_3$, we have
	\bse
	\wh T_{3,k}-T_{3,\infty}=\int_0^{\wt N}\frac{\wh Q_{n,k}(u)-Q_\infty(u)}{\wh G_{n,k,N}(u)}d\wh\Lambda_{n,k,C_N}(u)
	+\int_0^{\wt N}Q_\infty(u)\left\{\frac1{\wh G_{n,k,N}(u)}d\wh\Lambda_{n,k,C_N}(u)-\frac1{G_{\infty,N}(u)}d\Lambda_{\infty,C_N}(u)\right\}.
	\ese
	For the first integral above, by Condition \ref{assump:c2},
	\bse
	\E\left\{\int_0^{\wt N}\frac{\wh Q_{n,k}(u)-Q_\infty(u)}{\wh G_{n,k,N}(u)}d\wh\Lambda_{n,k,C_N}(u)\right\}^2=O_p\left[\E\left\{\sup_{u\in[0,\tau_N]}\left|\wh Q_{n,k}(u)-Q_{\infty}(u)\right|\right\}^2\right]=o_p(1).
	\ese
	For the second integral, since $Q_\infty$ is bounded, we have
	\bse
	&&\E\left[\int_0^{\wt N}Q_\infty(u)\left\{\frac1{\wh G_{n,k,N}(u)}d\wh\Lambda_{n,k,C_N}(u)-\frac1{G_{\infty,N}(u)}d\Lambda_{\infty,C_N}(u)\right\}\right]^2\\
    &\leq&
    \|Q_\infty\|_\infty^2\E\left[\int_0^{\wt N}
    \left|
    d\left\{\frac1{\wh G_{n,k,N}(u)}-\frac1{G_{\infty,N}(u)}\right\}
    \right|
    \right]^2 \\
	&=&O_p\left[\E\left\{\sup_{u\in[0,\tau_N]}\left|\frac{1}{\wh G_{n,k,N}(u)}-\frac{1}{G_{\infty,N}(u)}\right|\right\}^2\right]\\
	&=&o_p(1).
	\ese
	Therefore,
	\bse
	\max_k\|\widehat T_{3,k}-T_{3,\infty}\|_2=o_p(1).
	\ese
	
	Combining the three bounds gives
	\bse
	\max_k\|\widehat\phi^{\mathrm{obs}}_{n,k,a_0,t_0}-\phi^{\mathrm{obs}}_{\infty,a_0,t_0}\|_2
	\le\max_k\|\widehat T_{1,k}-T_{1,\infty}\|_2+\max_k\|\widehat T_{2,k}-T_{2,\infty}\|_2+\max_k\|\widehat T_{3,k}-T_{3,\infty}\|_2=o_p(1).
	\ese
	This proves the result.
\end{proof}

\begin{proof}[Proof of Theorem \ref{th:MR}]
	Let
	\bse
	\wh S_{n,a_0}(t_0)-S_{a_0}(t_0)=R_{1n}+R_{2n}+R_{3n},
	\ese
	where
	\bse
	R_{1n}&=&(\P_n-\E)\phi^{\mathrm{obs}}_{\infty,a_0,t_0}\,\\
	R_{2n}&=&\frac1K\sum_{k=1}^K\frac{Kn_k^{1/2}}{n}n_k^{1/2}(\P_{n,k}-\E)(\wh\phi^{\mathrm{obs}}_{n,k,a_0,t_0}-\phi^{\mathrm{obs}}_{\infty,a_0,t_0}),\\
	R_{3n}&=&\frac1K\sum_{k=1}^K\frac{Kn_k}{n}
	\E(\wh\phi^{\mathrm{obs}}_{n,k,a_0,t_0}-\phi^{\mathrm{obs}}_{\infty,a_0,t_0})+\left\{\E\phi^{\mathrm{obs}}_{\infty,a_0,t_0}-S_{a_0}(t_0)\right\},
	\ese
    where $\P_{n,k}$ denotes the empirical probability measure on set $\mathcal V_{n,k}$.
	
	First, by the weak law of large number,
	\bse
	R_{1n}=o_p(1).
	\ese
	
	As for $R_{2n}$, let $\mathbb{G}_{n,k}\equiv n_k^{1/2}(\P_{n,k}-\E)$, we have
	\bse
	\E\left|\mathbb G_{n,k}(\wh\phi^{\mathrm{obs}}_{n,k,a_0,t_0}-\phi^{\mathrm{obs}}_{\infty,a_0,t_0})\right|=\E\left[\E\left\{\left|\mathbb G_{n,k}(\wh\phi^{\mathrm{obs}}_{n,k,a_0,t_0}-\phi^{\mathrm{obs}}_{\infty,a_0,t_0})\right|\mid \mathcal T_{n,k}\right\}\right]=\E\left[\E\left\{\sup_{f\in \mathcal F_{n,k,t_0,a_0}}\left|\mathbb G_{n,k}f\right|\mid \mathcal T_{n,k}\right\}\right],
	\ese
	with $\mathcal F_{n,k,t_0,a_0}$ denoting the class of functions containing $\wh\phi^{\mathrm{obs}}_{n,k,a_0,t_0}-\phi^{\mathrm{obs}}_{\infty,a_0,t_0}$, which is a singleton class because $\wh\phi^{\mathrm{obs}}_{n,k,a_0,t_0}$ is a fixed function when conditioning on the training set $\mathcal T_{n,k}$. By Theorem 2.14.1 of \cite{van1996weak},
	there is a universal constant $C'$ such that
	\bse
	\E\left[\E\left\{\sup_{f\in \mathcal F_{n,k,t_0,a_0}}\left|\mathbb G_{n,k}f\right|\mid \mathcal T_{n,k}\right\}\right]\le C'\E\left[\E\left\{(\wh\phi^{\mathrm{obs}}_{n,k,a_0,t_0}-\phi^{\mathrm{obs}}_{\infty,a_0,t_0})^2\mid \mathcal T_{n,k}\right\}^{1/2}\right].
	\ese
	By Jensen's inequality, this is bounded by
	\bse
	C'\E\left[\E\left\{(\wh\phi^{\mathrm{obs}}_{n,k,a_0,t_0}-\phi^{\mathrm{obs}}_{\infty,a_0,t_0})^2\mid \mathcal T_{n,k}\right\}\right]^{1/2}=C'\left\{\E(\wh\phi^{\mathrm{obs}}_{n,k,a_0,t_0}-\phi^{\mathrm{obs}}_{\infty,a_0,t_0})^2\right\}^{1/2}.
	\ese
	In addition, 
	\bse
	\frac{Kn_k^{1/2}}{n}\le\frac{K(|n_k-n/K|+n/K)^{1/2}}{n}\le\frac{K|n_k-n/K|^{1/2}+K(n/K)^{1/2}}{n}\le \frac{K}{n}+\left(\frac{K}{n}\right)^{1/2}
	\ese
	for all $k$ since $|n_k-n/K|\le 1$ by assumption.
	Thus, together with Lemma \ref{lem:eif-consistency}, we have
	\bse
	\left|\frac1K\sum_{k=1}^K\frac{Kn_k^{1/2}}{n}\mathbb G_{n,k}(\wh\phi^{\mathrm{obs}}_{n,k,a_0,t_0}-\phi^{\mathrm{obs}}_{\infty,a_0,t_0})\right|\le O_p(n^{-1/2})\cdot o_p(1)=o_p(n^{-1/2}).	
	\ese
	That is, $R_{2n}=o_p(n^{-1/2})$.

	For the third term $R_{3n}$, by Cauchy-Schwarz inequality,
	\bse
	\left|\E(\wh\phi^{\mathrm{obs}}_{n,k,a_0,t_0}-\phi^{\mathrm{obs}}_{\infty,a_0,t_0})\right|\le
	\left\|\wh\phi^{\mathrm{obs}}_{n,k,a_0,t_0}-\phi^{\mathrm{obs}}_{\infty,a_0,t_0}\right\|_2
	=o_p(1),
	\ese
	again by Lemma \ref{lem:eif-consistency}. Hence,
	\bse
	\frac1K\sum_{k=1}^K\frac{Kn_k}{n}
	\E(\wh\phi^{\mathrm{obs}}_{n,k,a_0,t_0}-\phi^{\mathrm{obs}}_{\infty,a_0,t_0})
	=o_p(1).
	\ese
	
	Under Condition \ref{assump:c3}, the limiting nuisance functions satisfy one of the multiple robustness regimes. Therefore, by Lemma \ref{lem:MR},
	\bse
	\E\phi^{\mathrm{obs}}_{\infty,a_0,t_0}=S_{a_0}(t_0).
	\ese
	It follows that
	\bse
	R_{3n}=o_p(1).
	\ese
	
	Combining the three terms,
	\bse
	\widehat S_{n,a_0}(t_0)-S_{a_0}(t_0)
	=
	R_{1n}+R_{2n}+R_{3n}
	=
	o_p(1).
	\ese
	This proves the result.
\end{proof}

\subsection{Proof of Theorem \ref{Th:AL}}\label{prf:3}

\begin{proof}
	With the same decomposition as in the proof of Theorem \ref{th:MR},
	\bse
	\wh S_{n,a_0}(t_0)-S_{a_0}(t_0)=R_{1n}+R_{2n}+R_{3n},
	\ese
	where
	\bse
	R_{1n}&=&(\P_n-\E)\phi^{\mathrm{obs}}_{a_0,t_0},\\
	R_{2n}&=&\frac1K\sum_{k=1}^K\frac{Kn_k}{n}(\P_{n,k}-\E)(\wh\phi^{\mathrm{obs}}_{n,k,a_0,t_0}-\phi^{\mathrm{obs}}_{a_0,t_0}),\\
	R_{3n}&=&\frac1K\sum_{k=1}^K\frac{Kn_k}{n}
	\E(\wh\phi^{\mathrm{obs}}_{n,k,a_0,t_0}-\phi^{\mathrm{obs}}_{a_0,t_0}).
	\ese
	
	We have shown in the proof of Theorem \ref{th:MR} that
	\bse
	R_{2n}=o_p(n^{-1/2}).
	\ese
	
    For the third term $R_{3n}$, we analyze separately the three components
    $\widehat T_{1,k}$, $\widehat T_{2,k}$, and $\widehat T_{3,k}$, obtained from the corresponding terms $T_1$, $T_2$, and $T_3$ in the decomposition \eqref{eq:phi_decompose} by replacing the nuisance functions with their cross-fitted estimators in the observed-data EIF $\widehat\phi^{\mathrm{obs}}_{n,k,a_0,t_0}$.
	
	To analyze $\wh T_{1,k}$, define $\wh W_{n,k,N}(u\mid a,x,\wt{y},\delta_y)\equiv 
	G_{0N}(u\mid a,x,\wt y,\delta_y)/\wh G_{n,k,N}(u\mid a,x,\wt y,\delta_y)$,
	then
	\bse
	 \E(\wh T_{1,k})=\E\left\{\frac{\Delta_N}{\wh G_{n,k,N}(\wt N\mid a_0,X,\wt Y,\Delta_Y)}
	\wh\phi^F_{n,k,a_0,t_0}\right\}=
	\E\left\{\wh W_{n,k,N}(N\mid a_0,X,\wt Y,\Delta_Y)\wh \phi^F_{n,k,a_0,t_0}\right\}.
	\ese
	Using $\wh W_{n,k,N}(N\mid a_0,X,\wt Y,\Delta_Y)=1+\int_0^{\infty} I(N\ge u)d\wh W_{n,k,N}(u\mid a_0,X,\wt Y,\Delta_Y)$,
	we obtain
	\bse
	\E\left\{\frac{\Delta_N}{\wh G_{n,k,N}(\wt N\mid a_0,X,\wt Y,\Delta_Y)}
	\wh\phi^F_{n,k,a_0,t_0}\right\}&=&
	\E\left\{\int_0^N \wh\phi^F_{n,k,a_0,t_0} d\wh W_{n,k,N}(u\mid a_0,X,\wt Y,\Delta_Y)\right\}+\E(\wh\phi^F_{n,k,a_0,t_0})\\
	&=&\E\left\{\int_0^N \wh\phi^F_{n,k,a_0,t_0} d\wh W_{n,k,N}(u\mid a_0,X,\wt Y,\Delta_Y)\right\}+S_{a_0}(t_0)+o_p(n^{-1/2}),
	\ese
	where the second equality above is by Theorem 3 of \cite{Westling_2023}.
	
	Next,
	\bse
	&&\E(\wh T_{2,k})-\E(\wh T_{3,k})\\
    &=&\E\left\{
	\frac{1-\Delta_N}{\wh G_{n,k,N}(\wt N\mid a_0,X,\wt Y,\Delta_Y)}
	\wh Q_{n,k}(\wt N\mid  a_0,X,\wt Y,\Delta_Y)\right\}\\
    &&-
	\E\left\{\int_0^{\wt N}\frac{\wh Q_{n,k}(u\mid  a_0,X,\wt Y,\Delta_Y)}{\wh G_{n,k,N}(u\mid  a_0,X,\wt Y,\Delta_Y)}
	d\wh\Lambda_{n,k,C_N}(u\mid  a_0,X,\wt Y,\Delta_Y)
	\right\},
	\ese
    where $\wh Q_{n,k}(u\mid a,x,\wt y,\delta_y)=
Q(u\mid a,x,\wt y,\delta_y;\wh\eta_{n,k})$.
	
	By 
	\bse 
	d\wh W_{n,k,N}(u\mid  a_0,x,\wt y,\delta_y)=\wh W_{n,k,N}(u\mid  a_0,x,\wt y,\delta_y)\{d\wh\Lambda_{n,k,C_N}(u\mid  a_0,x,\wt y,\delta_y)-d\Lambda_{0C_N}(u\mid  a_0,x,\wt y,\delta_y)\},
	\ese
	we have
	\bse
    \E(\wh T_{2,k})-\E(\wh T_{3,k})=
	-\E\left\{\int_0^N\wh Q_{n,k}(u\mid a_0,X,\wt Y,\Delta_Y)d\wh W_{n,k,N}(u\mid a_0,X,\wt{Y},\Delta_Y)\right\}.
	\ese
	
	Combining with the expansion of the $E(\wh T_{1,k})$, we have
	\be\label{eq:123-diff}
    &&\E(\wh T_{1,k})+\E(\wh T_{2,k})-\E(\wh T_{3,k})-S_{a_0}(t_0)+o_p(n^{-1/2})\nonumber\\
	&=& \E\left\{\int_0^N \wh \phi^F_{n,k,a_0,t_0} d\wh W_{n,k,N}(u\mid a_0,X,\wt{Y},\Delta_Y)\right\}- \E\left\{\int_0^N\wh Q_{n,k}(u\mid a_0,X,\wt{Y},\Delta_Y)d\wh W_{n,k,N}(u\mid a_0,X,\wt{Y},\Delta_Y)
	\right\}\nonumber\\ 
	&=& \E\left[\int_0^N\left\{\wh\phi^F_{n,k,a_0,t_0}-Q_k(u)\right\}d\wh W_{n,k,N}(u)\right]+\E\left[\int_0^N\left\{Q_k(u)-\wh Q_{n,k}(u)\right\}d\wh W_{n,k,N}(u)\right],
	\ee
	where $Q_k(u\mid a,x,\wt y,\delta_y)= \E(\wh\phi^F_{n,k,a_0,t_0}\mid N\ge u,a,x,\wt y,\delta_y)$.
	Since 
	\bse
	&& \E\left[\int_0^N\left\{\wh\phi^F_{n,k,a_0,t_0}-Q_k(u)\right\}d\wh W_{n,k,N}(u)\mid a_0,X,\wt Y,\Delta_Y\right]\\
	&=&\int_0^\infty  \E\left[I(N\ge u)\left\{\wh\phi^F_{n,k,a_0,t_0}-Q_k(u)\right\}\mid a_0,X,\wt Y,\Delta_Y\right]d\wh W_{n,k,N}(u)\\
	&=&\int_0^\infty \left[P(N\ge u\mid a_0,X,\wt Y,\Delta_Y)Q_k(u)- \E\left\{I(N\ge u)Q_k(u)\mid a_0,X,\wt Y,\Delta_Y\right\}\right]d\wh W_{n,k,N}(u)\\
	&=&0,
	\ese
	equation \eqref{eq:123-diff} reduces to
	\bse
    &&\E(\wh T_{1,k})+\E(\wh T_{2,k})-\E(\wh T_{3,k})-S_{a_0}(t_0)+o_p(n^{-1/2})\\
    &=&\E\left[\int_0^N\left\{Q_k(u)-\wh Q_{n,k}(u)\right\}d\wh W_{n,k,N}(u)\right]\\
    &=&\E\left[\int_0^N\left\{Q_k(u)-\wh Q_{n,k}(u)\right\}\wh W_{n,k,N}(u)\left\{d\wh\Lambda_{n,k,C_N}(u)-d\Lambda_{0C_N}.(u)\right\}\right]
	\ese
	By Condition \ref{assump:c2}, we have that $\wh W_{n,k,N}$ is bounded. In addition, by condition \ref{assump:c4}, we have
	\bse
	\E\left[\int_0^N\{Q_k(u)-\wh Q_{n,k}(u)\}d\wh W_{n,k,N}(u)\right]=o_p(n^{-1/2}).
	\ese

	Finally, we obtain
	\bse
	R_{3n}=S_{a_0}(t_0)+o_p(n^{-1/2})-\E\phi^{\mathrm{obs}}_{a_0,t_0}=o_p(n^{-1/2}).
	\ese
	
	Therefore,
	\bse
	\wh S_{n,a_0}(t_0)-S_{a_0}(t_0)=
	( \P_n-\mathbb E)\phi^{\mathrm{obs}}_{a_0,t_0}+o_p(n^{-1/2}).
	\ese
	
	Multiplying by $\sqrt n$,
	\bse
	\sqrt{n}\left\{\wh S_{n,a_0}(t_0)-S_{a_0}(t_0)\right\}=\frac1{\sqrt n}
	\sum_{i=1}^n\left\{\phi^{\mathrm{obs}}_{a_0,t_0}(O_i^{\mathrm{obs}})-S_{a_0}(t_0)\right\}+o_p(1).
	\ese
	
	The result follows from the central limit theorem.
\end{proof}

\end{document}